\documentclass{article}

\usepackage[preprint]{neurips_2026}

\usepackage[utf8]{inputenc} 
\usepackage[T1]{fontenc}    
\usepackage{hyperref}       
\usepackage{url}            
\usepackage{booktabs}       
\usepackage{amsfonts}       
\usepackage{nicefrac}       
\usepackage{microtype}      
\usepackage{xcolor}         
\usepackage{graphicx}
\usepackage{bm}
\usepackage{threeparttable}
\usepackage{sidecap}
\sidecaptionvpos{figure}{c}
\usepackage{tikz}
\usetikzlibrary{decorations.pathreplacing}
\usepackage{quantikz}
\usepackage{amsthm}
\usepackage{amsmath}

\newtheorem{theorem}{Theorem}[section]
\newtheorem{corollary}{Corollary}[theorem]
\newtheorem{lemma}[theorem]{Lemma}
\newtheorem{proposition}[theorem]{Proposition}

\theoremstyle{definition}
\newtheorem{definition}{Definition}[section]
\newtheorem{remark}{Remark}[section]

\title{Architecture Shape Governs QNN Trainability: 
Jacobian Null Space Growth and Parameter Efficiency}

\author{%
  \normalfont
  Michael Poppel$^{1,2}$ \quad
  David Bucher$^{2}$ \quad
  Maximilian Zorn$^{1}$ \quad
  Markus Baumann$^{1}$ \quad
  Sebastian W\"{o}lckert$^{1}$ \\
  Claudia Linnhoff-Popien$^{1}$ \quad
  Philipp Altmann$^{1}$ \quad
  Jonas Stein$^{1}$ \\[0.5em]
  $^{1}$Department of Computer Science, LMU Munich, Germany \\
  $^{2}$Aqarios GmbH, Munich, Germany \\[0.3em]
  \texttt{michael.poppel@aqarios.com}
}

\begin{document}

\maketitle

\begin{abstract}
Variational quantum circuits with angle encoding implement truncated
Fourier series, and architectures arranging $N$ qubits with $L$
encoding layers each --- sharing encoding budget $E = NL$ ---
generate identical frequency spectra, identical frequency redundancy,
and require the same minimum parameter count for coefficient control.
Despite this equivalence, trainability varies substantially
with architecture shape $(N,L)$ at fixed $E$.
We identify structural rank deficiency of the coefficient matching
Jacobian $J$ as the mechanism responsible.
For serial single-qubit architectures, we prove $\mathrm{rank}(J)
\leq 2L+1$ regardless of parameter count $P$, with $\dim(\ker J) \geq P-(2L+1)$ growing without bound --- a
phenomenon we term \emph{structural gradient starvation}: a
growing fraction of parameters become structurally decoupled
from the loss as $P$ increases at fixed $L$.
Parallel architectures avoid this via independent phase
trajectories, ensuring $\sigma_{\min}(J^{(\mathrm{par})}) > 0$
generically for $P \leq 2E+1$, so no parameter lies in $\ker J$.
For practitioners, we further show that the two natural routes to
increasing parameter count have fundamentally different effects:
adding feature map (FM) layers monotonically strengthens the
Jacobian QFIM eigenvalue spectrum and achieves $R^2 \geq 0.95$
with $1.6$--$2.2\times$ fewer parameters than adding trainable
blocks across all tested architectures, while trainable blocks
improve training only through the classical interpolation
mechanism with no quantum-specific benefit.
\end{abstract}

\section{Introduction}

Variational quantum circuits (VQCs) with angle encoding implement
truncated Fourier series~\citep{schuld2021effect}, and are universal
function approximators even for a single qubit~\citep{Perez_Salinas_2021}.
Arranging $N$ qubits with $L$ encoding layers each, the accessible
frequency spectrum depends only on the total encoding budget $E = NL$,
not on how it is distributed across $(N,L)$ at fixed
$E$~\citep{holzer2024spectral}.
Under unary encoding, frequency redundancy --- the spectral bias
toward lower frequencies~\citep{Rahaman_2019} arising from
combinatorial degeneracy in the encoding layers --- is likewise
governed entirely by $E$~\citep{duffy2026spectralbiasvariationalquantum}.
Both the accessible function class and the spectral bias mechanism
are therefore architecture-independent.
\citet{holzer2024spectral} observe that larger $N$ tends to yield
more favorable optimization behavior, and note that a large
frequency spectrum can make accurate training harder even when
theoretical expressivity is sufficient.
Figure~\ref{fig:intro} makes this precise: at fixed $E$ --- where
all architectures share identical spectra and redundancy structure
--- low-qubit ($N=1,2$) architectures fail across the full
parameter range while intermediate architectures ($N=3,4$) succeed
with far fewer parameters.
Since single-qubit circuits are universal function approximators
in the limit~\citep{Perez_Salinas_2021}, expressivity alone cannot
explain this failure.
\begin{figure}[h]
  \centering
  \includegraphics[width=\linewidth]{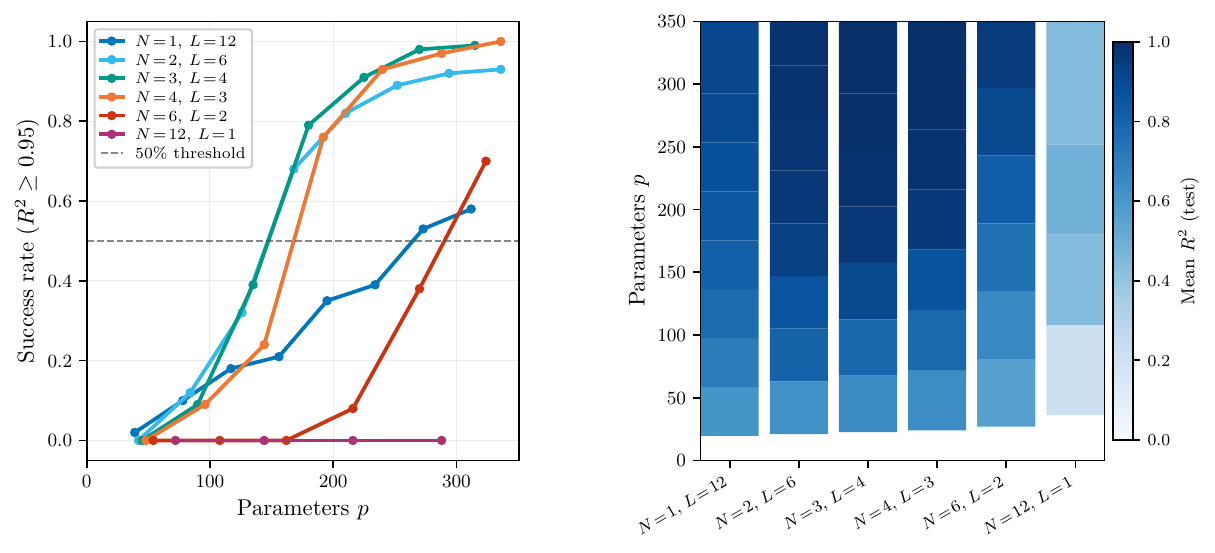}
  \caption{
  Trainability of data re-uploading QNNs at fixed encoding budget
  $E=NL=12$, trained with Adam on 10 randomly drawn degree-12
  Fourier targets.
  \textbf{Left:} success rate ($R^2 \geq 0.95$) vs.\ parameter
  count $P$.
  \textbf{Right:} mean $R^2$ (test) vs.\ $P$.
  All architectures share identical frequency spectra and redundancy
  structure, ruling out expressivity and spectral bias as
  explanations.
  }
  \label{fig:intro}
\end{figure}
We identify \textbf{structural rank deficiency of the Jacobian $J$}
as the mechanism responsible.
For serial ($N=1$) circuits, phase-locking forces all $P$ gradient
directions to share a common global phase, giving
$\mathrm{rank}(J) \leq 2L+1$ regardless of $P$ --- a property
we term \emph{structural gradient starvation} (distinct from
gradient suppression in deep learning~\citep{glorot2010understanding,
chen2021exploring}): the fraction of parameters receiving identically
zero gradient signal grows as $(P-\mathrm{rank}(J))/P \to 1$.
Parallel architectures escape this via independent phase trajectories
per qubit, maintaining $\sigma_{\min}(J^{(\mathrm{par})}) > 0$
generically for $P \leq 2E+1$; the advantage is structural rather
than threshold-based since both architectures share the same null-space
onset $P > 2E+1$ at fixed $E$.
We make two contributions
(Propositions~\ref{prop:rank}--\ref{prop:ill_cond},
Corollary~\ref{cor:parallel}, Section~\ref{sec:theory}):
a proof of the rank ceiling and null-space growth for serial circuits,
and a proof that parallel architectures avoid phase-locking via
bilinear factorization.
We further show that adding FM layers achieves $R^2 \geq 0.95$
with $1.6$--$2.2\times$ fewer parameters than adding trainable blocks, \emph{while}
trainable blocks improve training only through the classical
interpolation mechanism (Section~\ref{sec:op}).

\section{Theoretical Background}
\label{sec:background}

\subsection{Variational Quantum Circuits with Angle Encoding}
\label{sec:vqc}

\begin{definition}[VQC with Angle Encoding]
\label{def:vqc}
An $L$-layer VQC with angle encoding consists of feature map (FM)
operations $S(x) = e^{ixH}$, where $H = \frac{1}{2}\sigma$ for a Pauli
operator $\sigma \in \{\sigma_x, \sigma_y, \sigma_z\}$, interleaved
with parametrized ansatz layers $W_i(\boldsymbol{\theta}_i)$.
The circuit unitary and scalar output are
\begin{equation}
U(x,\boldsymbol{\theta}) =
W_L(\boldsymbol{\theta}_L)S(x)\cdots S(x)W_0(\boldsymbol{\theta}_0),
\qquad
f_{\boldsymbol{\theta}}(x) =
\langle 0 | U^{\dagger}(x,\boldsymbol{\theta})\,M\,
U(x,\boldsymbol{\theta}) | 0 \rangle,
\end{equation}
where $M$ is a fixed observable and
$\boldsymbol{\theta} = (\boldsymbol{\theta}_0, \ldots,
\boldsymbol{\theta}_L)$ are optimized by minimizing a training loss
$\mathcal{L}(\boldsymbol{\theta})$.
Note that an $L$-layer VQC contains $L+1$ trainable ansatz blocks
$W_0, \ldots, W_L$, sandwiching $L$ encoding layers $S(x)$;
parameter counts throughout this paper reflect this convention.
\end{definition}

\subsection{Fourier Representation, Spectral Invariance, and Spectral Bias}
\label{sec:fourier}

\begin{theorem}[Fourier Representation~\citep{schuld2021effect}]
\label{thm:fourier}
The output of a VQC with $L$ FM layers is a truncated Fourier series:
\begin{equation}
f_{\boldsymbol{\theta}}(x) =
\sum_{\omega \in \Omega} c_{\omega}(\boldsymbol{\theta})\,e^{i\omega x},
\end{equation}
with parameter-dependent coefficients $c_\omega(\boldsymbol{\theta})
\in \mathbb{C}$.
Under \emph{unary encoding} --- where each FM layer applies
$S(x) = e^{ix\sigma/2}$ with unit scaling --- the spectrum is
$\Omega = \{-L, \ldots, L\}$ with $|\Omega| = 2L+1$.
\end{theorem}

\begin{corollary}[Parameter Requirement~\citep{schuld2021effect}]
\label{cor:param_scaling}
Independent control of all Fourier coefficients requires
$P \geq 2L + 1$ ansatz parameters.
\end{corollary}

\begin{theorem}[Spectral Invariance~\citep{holzer2024spectral}]
\label{thm:spectral_invariance}
For $N$ qubits with $L$ FM layers each, the spectrum depends only on
the total encoding budget $E = NL$: serial ($N=1$, $L=E$), parallel
($N=E$, $L=1$), and hybrid architectures with entangling ansatz layers
all generate identical frequency spectra $\Omega = \{-E, \ldots, E\}$.
Any trainability difference across architectures at fixed $E$ is
therefore a property of the ansatz structure, not of the function class.
\end{theorem}

Under unary encoding, each frequency $\omega$ can be generated by
$\binom{2E}{E-\omega}$ distinct FM combinations, a
degeneracy that grows with $E$ and introduces a spectral bias
toward lower frequencies~\citep{duffy2026spectralbiasvariationalquantum},
analogous to that in deep learning~\citep{Rahaman_2019,Belis_2026}.
This redundancy depends only on the total encoding budget $E$,
and is therefore identical across architectures at fixed $E$.

\subsection{Quantum Fisher Information Matrix}
\label{sec:dla_qfim}

\begin{definition}[Quantum Fisher Information Matrix~\citep{stokes2020quantum}]
\label{def:qfim}
For $|\psi(\boldsymbol{\theta})\rangle = U(\boldsymbol{\theta})|0\rangle$,
the QFIM $\mathcal{F}(\boldsymbol{\theta}) \in \mathbb{R}^{P\times P}$
has entries
\begin{equation}
\mathcal{F}_{jk} =
4\,\mathrm{Re}\!\left[
\langle \partial_j \psi | \partial_k \psi \rangle
- \langle \partial_j \psi | \psi \rangle
\langle \psi | \partial_k \psi \rangle
\right],
\end{equation}
with rank equal to the number of independently explorable directions
in state space at $\boldsymbol{\theta}$.
\end{definition}

The maximum achievable QFIM rank is bounded above by the real
dimension of the pure state manifold:
\begin{equation}
  R_{\max} := \max_{\boldsymbol{\theta}}\,
  \mathrm{rank}(\mathcal{F}(\boldsymbol{\theta}))
  \leq 2^{N+1}-2,
\end{equation}
independently of $P$~\citep{larocca2023theory}.
A VQC with $P \ll R_{\max}$ parameters operates in the
underparameterized regime where the optimization landscape
may contain spurious local minima; $P \geq R_{\max}$ is a
necessary condition for landscape tractability in
theory~\citep{larocca2023theory}, though empirically circuits
often train successfully with fewer parameters.
We invoke this bound in Appendix~\ref{app:degree_justification}
to explain why the $N=6$ minimum configuration ($P=36$,
$R_{\max}=126$) starts well below saturation.

Our primary diagnostic --- the Jacobian QFIM
$\hat{\mathcal{F}} = \mathbf{J}_{\mathrm{data}}^\top
\mathbf{J}_{\mathrm{data}} / n$
(Definition~\ref{def:jacobian} and
Section~\ref{sec:coeff_matching}) --- measures output
sensitivity at discrete training points and has different
rank bounds from $\mathcal{F}(\boldsymbol{\theta})$;
it is the appropriate diagnostic for the function
approximation setting studied here.

\subsection{The Coefficient Matching Problem}
\label{sec:coeff_matching}

By Theorem~\ref{thm:fourier}, training a VQC reduces to finding
$\boldsymbol{\theta}^*$ such that $c_\omega(\boldsymbol{\theta}^*)
= c^*_\omega$ for all $\omega \in \Omega$ --- the
\emph{coefficient matching problem}.
We assume the expressivity condition
$\Omega_{\mathrm{target}} \subseteq \Omega_{\mathrm{model}}$
is satisfied throughout.

\begin{definition}[Coefficient Matching Jacobian]
\label{def:jacobian}
Although $c_\omega(\boldsymbol{\theta}) \in \mathbb{C}$
(Theorem~\ref{thm:fourier}), the circuit output
$f_{\boldsymbol{\theta}}(x) \in \mathbb{R}$ implies conjugate
symmetry $c_{-\omega} = \overline{c_{\omega}}$, leaving
$|\Omega| = 2L+1$ real degrees of freedom.
We work with the \emph{real representative} obtained by stacking
independent coefficients:
\begin{equation}
  \tilde{\mathbf{c}}(\boldsymbol{\theta})
  = \bigl(\mathrm{Re}\,c_0,\;
    \mathrm{Re}\,c_1,\;\mathrm{Im}\,c_1,\;\ldots,\;
    \mathrm{Re}\,c_L,\;\mathrm{Im}\,c_L
  \bigr)^\top \in \mathbb{R}^{2L+1}.
\end{equation}
The \emph{coefficient matching Jacobian} is then
\begin{equation}
J \in \mathbb{R}^{|\Omega| \times P},
\qquad
J_{k,j} =
\frac{\partial \tilde{c}_k(\boldsymbol{\theta})}
     {\partial \theta_j},
\end{equation}
with rows indexed by $k \in \{1,\ldots,2L+1\}$ corresponding
to the real representatives above.
Its rank governs the number of Fourier directions accessible
to the optimizer: parameters in $\ker J$ are structurally
decoupled from the loss and receive identically zero gradient signal.
\end{definition}

The \emph{data Jacobian}
$\mathbf{J}_{\mathrm{data}} \in \mathbb{R}^{n \times P}$,
with entries
$(\mathbf{J}_{\mathrm{data}})_{i,j} =
 \partial f_{\boldsymbol{\theta}}(x_i)/\partial\theta_j$
evaluated at $n$ training points $\{x_i\}_{i=1}^n$,
measures sensitivity of the scalar circuit output over the
training distribution.
The \emph{Jacobian QFIM}
$\hat{\mathcal{F}} = \mathbf{J}_{\mathrm{data}}^\top
\mathbf{J}_{\mathrm{data}} / n \in \mathbb{R}^{P \times P}$
is distinct from the coefficient Jacobian $J$: $J$ governs
the geometry of coefficient space and is the object studied
in Sections~\ref{sec:theory}--\ref{sec:op}, while
$\hat{\mathcal{F}}$ measures output sensitivity at discrete
training points and serves as the empirical diagnostic
in Section~\ref{sec:op}.
When $\mathrm{rank}(J) < P$, parameters in $\ker J$ produce no
change in any Fourier coefficient and receive identically zero
gradient signal regardless of initialization or training
trajectory~\citep{nocedal2006numerical,trefethen1997numerical,
bjorck1996numerical} --- a structural phenomenon distinct from barren
plateaus~\citep{mcclean2018barren}.

By Theorem~\ref{thm:spectral_invariance}, serial and parallel
architectures at fixed $E$ share the same $|\Omega|$ and hence
the same dimensions of $J$; any difference in $\dim(\ker J)$
is purely structural, arising from how parameters couple to the
Fourier coefficients --- the subject of Section~\ref{sec:theory}.

\section{Structural Gradient Starvation in Serial Architectures}
\label{sec:theory}

We now identify the structural property that distinguishes
serial from parallel architectures at fixed encoding budget
$E = NL$, characterizing it through a concrete comparison of
the coefficient matching equations before proving the resulting
rank and null-space bounds.

\subsection{Structural Comparison: Serial vs.\ Parallel}
\label{subsec:concrete}

We illustrate the structural difference at the smallest
non-trivial encoding budget $E=2$, where both architectures
share spectrum $\Omega=\{-2,-1,0,1,2\}$.

\paragraph{Serial ($N=1$, $L=2$).}
With three SU(2) parameter matrices $W^{(1)},W^{(2)},W^{(3)}$,
writing $W^{(l)} = \bigl(\begin{smallmatrix}\alpha_l & -\bar\beta_l \\
\beta_l & \bar\alpha_l\end{smallmatrix}\bigr)$ for $l=1,2,3$,
the highest-frequency coefficient is:
\begin{equation}
  c_{+2}^{\mathrm{ser}} = -2\,\alpha_1\,\alpha_2^2\,\bar\beta_1\,
  \beta_3\bar\alpha_3.
  \label{eq:ser_c2}
\end{equation}
All three matrices are coupled simultaneously; the interior
matrix $W^{(2)}$ appears as $\alpha_2^2$, re-entangling each
gradient update with the current parameter value.

\paragraph{Parallel ($N=2$, $L=1$).}
With $\mathbf{v}=W^{(1)}|0\rangle\in\mathbb{C}^4$ and
$\widetilde{M}=(W^{(2)})^\dagger M\,W^{(2)}$:
\begin{equation}
  c_{+2}^{\mathrm{par}} = v_1\bar v_4\cdot\widetilde{M}_{4,1}.
  \label{eq:par_c2}
\end{equation}
Every summand is bilinear: the $W^{(1)}$-factor and the
$W^{(2)}$-factor are completely separated.
This pattern holds for all frequencies and generalizes to
arbitrary $N$ (Proposition~\ref{prop:bilinear}).

\paragraph{The general pattern.}
For general $L$ and $N$, these structural differences persist: 
serial coefficients are degree-$2L+2$ polynomials in the ansatz 
parameters with all matrices coupled, while parallel coefficients 
remain bilinear with $W^{(1)}$ and $W^{(2)}$ fully separated for 
all frequencies and all $N$ (Proposition~\ref{prop:bilinear}, 
Appendix~\ref{app:bilinear}; see Appendix~\ref{app:bilinear_fig} 
for a circuit diagram illustration).

\begin{proposition}[Single-qubit gradient rank]
  \label{prop:rank}
  For a single-qubit PQC, at any fixed $(x,\boldsymbol{\theta})$,
  $\nabla_{\boldsymbol{\theta}} f(x;\boldsymbol{\theta})
  \in\mathbb{R}^P$ lies in a subspace of dimension at most $2$,
  and $\mathrm{rank}(\mathcal{F}(\boldsymbol{\theta}))\leq 2$.
  Both bounds are tight for generic $M$ and $|\psi\rangle$.
\end{proposition}

\begin{proof}[Proof sketch]
For $N=1$, $|\psi\rangle$ lives on $\mathbb{CP}^1\cong S^2$, real
dimension 2.
Its tangent space is spanned by a single complex vector
$|\psi^\perp\rangle$, so $|\partial_j\psi\rangle
= g_j|\psi^\perp\rangle + \lambda_j|\psi\rangle$.
The gradient component reduces to $2\,\mathrm{Re}[\bar g_j\mu]$
where $\mu:=\langle\psi^\perp|M|\psi\rangle\in\mathbb{C}$ is
\emph{shared across all $j$}.
Hence $\nabla_{\boldsymbol{\theta}} f
\in\mathrm{span}\{\mathrm{Re}(g),\mathrm{Im}(g)\}$,
dimension $\leq 2$.
Tightness holds outside a measure-zero set of parameters,
as shown in Appendix~\ref{app:prop1}.
\end{proof}

\subsection{Serial Phase-Locking}
\label{lemma:coupling_main}

Proposition~\ref{prop:rank} gives a 2D bound at each fixed $x$.
It is \emph{a priori} possible that the subspace
$\mathcal{V}(x,\boldsymbol{\theta})=
\mathrm{span}\{\mathrm{Re}(g(x)),\mathrm{Im}(g(x))\}$
rotates as $x$ varies and spans a larger subspace in $J$.
The serial product structure prevents this.

\begin{lemma}[Serial phase-locking]
  \label{lemma:serial_coupling}
  For a serial $N=1$ re-uploading circuit:
  \begin{equation}
    g_j(x) = h_j(x)\cdot G(x,\boldsymbol{\theta}),
    \label{eq:factored}
  \end{equation}
  where $G(x,\boldsymbol{\theta})\in\mathbb{C}$ is a global phase
  factor common to all parameters and $h_j(x)\in\mathbb{C}$
  depends only on the sub-circuit below layer $j$.
  The subspace $\mathcal{V}(x,\boldsymbol{\theta})$ rotates by
  the \emph{same phase for every parameter} as $x$ varies.
\end{lemma}

\begin{proof}[Proof sketch]
For each layer $j$, $D_j^\dagger|\psi^\perp\rangle = 
e^{i\varphi}|\hat\phi_j^\perp\rangle$ where the phase $\varphi$ 
is independent of $j$.
We define $|\psi^\perp\rangle := T|\psi\rangle$ where
$T = i\sigma_y K$ is the anti-unitary time-reversal operator,
giving a canonically defined orthogonal complement that depends
only on $|\psi\rangle$ and is shared across all $j$.
With this gauge, $\varphi$ depends only on the global state
$|\psi\rangle$, establishing the phase independence.
This yields $g_j = h_j \cdot e^{i\varphi}$, with $h_j$ a 
trigonometric polynomial of degree $\leq L$ determined by 
the sub-circuit up to and including layer $j$.
Full proof in Appendix~\ref{app:lemma}.
\end{proof}

Figure~\ref{fig:phase_locking} provides direct empirical
confirmation: gradient trajectories $g_j(x)$ from different
ansatz blocks trace curves of identical shape in the complex
plane for serial circuits ($\mathrm{Im}(g_j/g_k) = 0.000$
exactly), while parameters on different qubits in parallel circuits 
trace genuinely distinct trajectories 
($|\mathrm{Im}(g_j/g_k)| = 0.19$ on average),
as established in Corollary~\ref{cor:parallel} below.

\subsection{Structural Gradient Starvation in Serial
\texorpdfstring{$N=1$}{N=1} Circuits}
\label{prop:prop2_main}

\begin{proposition}[Structural gradient starvation in serial circuits]
  \label{prop:ill_cond}
  Fix $L$ and let $P$ grow.
  For a serial single-qubit architecture:
  \begin{enumerate}
    \item[(i)] $\mathrm{rank}(J) = 2L+1$ generically,
      so the rank ceiling imposed by the $2L+1$ rows of $J$
      is tight and cannot be improved by adding parameters.
    \item[(ii)] For $P > 2L+1$, $\dim(\ker J) \geq P-(2L+1)
      \to \infty$ as $P \to \infty$ at fixed $L$: the fraction
      of parameters receiving zero gradient signal,
      $(P-\mathrm{rank}(J))/P \to 1$, grows monotonically
      with $P$.
  \end{enumerate}
\end{proposition}

\begin{proof}[Proof sketch]
The bound $\mathrm{rank}(J) \leq 2L+1$ follows from
matrix dimensions since $J \in \mathbb{R}^{(2L+1) \times P}$.
Tightness --- $\mathrm{rank}(J) = 2L+1$ generically ---
follows from the phase-locking of
Lemma~\ref{lemma:serial_coupling}: each $h_j(x)$ is a
trigonometric polynomial of degree $\leq L$, so all
$\{h_j\}$ lie in
$\mathcal{T}_L = \mathrm{span}\{e^{ikx}:|k|\leq L\}$,
real dimension $2L+1$; for generic $\boldsymbol{\theta}$
the $2L+1$ rows of $J$ are linearly independent, confirming
that the ceiling is achieved.
For $P > 2L+1$, rank-nullity gives
$\dim(\ker J) \geq P-(2L+1) > 0$,
growing without bound as $P \to \infty$ at fixed $L$.
Every parameter in $\ker J$ produces no change in any Fourier
coefficient and receives identically zero gradient signal,
so the fraction of structurally decoupled parameters
$(P-\mathrm{rank}(J))/P \to 1$, giving~(ii).
Full proof in Appendix~\ref{app:prop2}.
\end{proof}

Proposition~\ref{prop:ill_cond} identifies $P = 2L+1$ as the
\emph{Jacobian-rank optimum} for serial circuits: below this
threshold the circuit cannot span all target Fourier directions;
above it, null space growth dominates and any further improvement
must come from the classical interpolation mechanism
(Section~\ref{sec:op}) rather than expanded Fourier coverage.
This structural gradient starvation arises from
the phase-locking of Lemma~\ref{lemma:serial_coupling}
regardless of how parameters are added.

\begin{corollary}[Parallel architecture: structural advantage via 
  independent phase trajectories]
  \label{cor:parallel}
  Consider a parallel $N$-qubit architecture with $L$ FM layers
  per qubit and encoding budget $E = NL$, with entangling ansatz
  layers and observable $M = Z_{N-1}$ as used in our experiments.

  \begin{enumerate}
    \item[(i)] \textbf{Onset threshold.}
      At fixed encoding budget $E = NL$, both serial and parallel
      architectures develop null directions in $J$ when $P > 2E+1$
      (Proposition~\ref{prop:ill_cond}), consistop\textbf{}ent with the shared
      frequency spectrum of
      Theorem~\ref{thm:spectral_invariance}.
      The trainability difference between architectures at fixed $E$
        therefore cannot be explained by a difference in onset threshold
        alone.

    \item[(ii)] \textbf{Absence of phase-locking.}
      By Lemma~\ref{lemma:serial_coupling}, serial gradient starvation arises because all $P$ parameters share
      a common global phase $G(x,\boldsymbol{\theta})=e^{i\varphi}$,
      confining $\mathrm{rank}(J^{(\mathrm{ser})}) \leq 2L+1$
      independently of $P$.
      For $N > 1$ qubits, each qubit's local reduced state
      evolves with an independent local phase $\varphi^{(n)}$,
      since the downstream operator $D_j^{(n)}$ and state
      $|\psi^{(n)}\rangle$ differ across qubits.
      Parameters on different qubits are therefore modulated
      by independent phases, breaking the global coherence
      that produces the serial rank ceiling.

    \item[(iii)] \textbf{Full column rank below the threshold.}
      Because phase-locking is absent, the column space of
      $J^{(\mathrm{par})}$ is not confined to a single
      $\mathcal{T}_L$-subspace.
      For $P \leq 2E+1$, $\sigma_{\min}(J^{(\mathrm{par})}) > 0$
        generically, so all parameters contribute to at least one
        Fourier direction and no parameter lies in $\ker J$.
        By contrast, for serial circuits
        $\dim(\ker J^{(\mathrm{ser})})/P \to 1$ as $P \to \infty$,
        so an ever-growing fraction of parameters are structurally
        decoupled from the loss across the same range.
  \end{enumerate}
  We validate the rank ceiling empirically in
Appendix~\ref{app:rank_saturation}
(Figure~\ref{fig:scaling_separation}); a complete
characterization of $\mathrm{rank}(J^{(\mathrm{par})})$
for entangling ans\"{a}tze remains an open problem.
\end{corollary}

\begin{remark}[Entangling ansätze]
\label{rem:entangling}
The parallel advantage of Corollary~\ref{cor:parallel} extends to
the Rot-CNOT ladder used in our experiments: CNOT gates are
unparameterized, so inter-qubit coupling does not grow with $P$,
and phase-locking depends on the $\mathbb{CP}^1$ geometry of a
single qubit and is broken by independent local phases regardless
of entanglement.
A full mechanistic justification is given in
Appendix~\ref{app:experiments}.
\end{remark}

\begin{proof}[Proof sketch]
For $N > 1$ qubits, each qubit's local reduced state evolves
with an independent local phase $\varphi^{(n)}$, since the
downstream operator $D_j^{(n)}$ and state $|\psi^{(n)}\rangle$
differ across qubits.
The global coherence condition of
Lemma~\ref{lemma:serial_coupling} is therefore broken: columns
of $J^{(\mathrm{par})}$ are modulated by independent phases
and are not confined to a single $\mathcal{T}_L$-subspace,
giving $\sigma_{\min}(J^{(\mathrm{par})}) > 0$ generically
for $P \leq 2E+1$ by a real-analyticity argument.
Null directions appear at $P > 2E+1$ by rank-nullity applied
to the shared spectrum $|\Omega| = 2E+1$, the same threshold
as the serial case at fixed $E$.
Full proof in Appendix~\ref{app:cor}.
\end{proof}

\paragraph{Summary.}
Phase-locking in serial circuits confines all $P$ gradient
directions to $\mathcal{T}_L$, as confirmed directly by
Figure~\ref{fig:phase_locking}.
Together, Proposition~\ref{prop:ill_cond} and
Corollary~\ref{cor:parallel} establish the key separation
(Theorem~\ref{thm:separation}, Appendix~\ref{app:separation}):
serial gradient starvation sets in at $P > 2L+1$ regardless of
parameter count, while parallel architectures maintain
$\sigma_{\min}(J^{(\mathrm{par})}) > 0$ generically for
$P \leq 2E+1$, deferring null space growth to the same threshold
$P > 2E+1$ that applies to serial circuits at fixed $E$.

\begin{figure}[t]
  \centering
  \includegraphics[width=\linewidth]{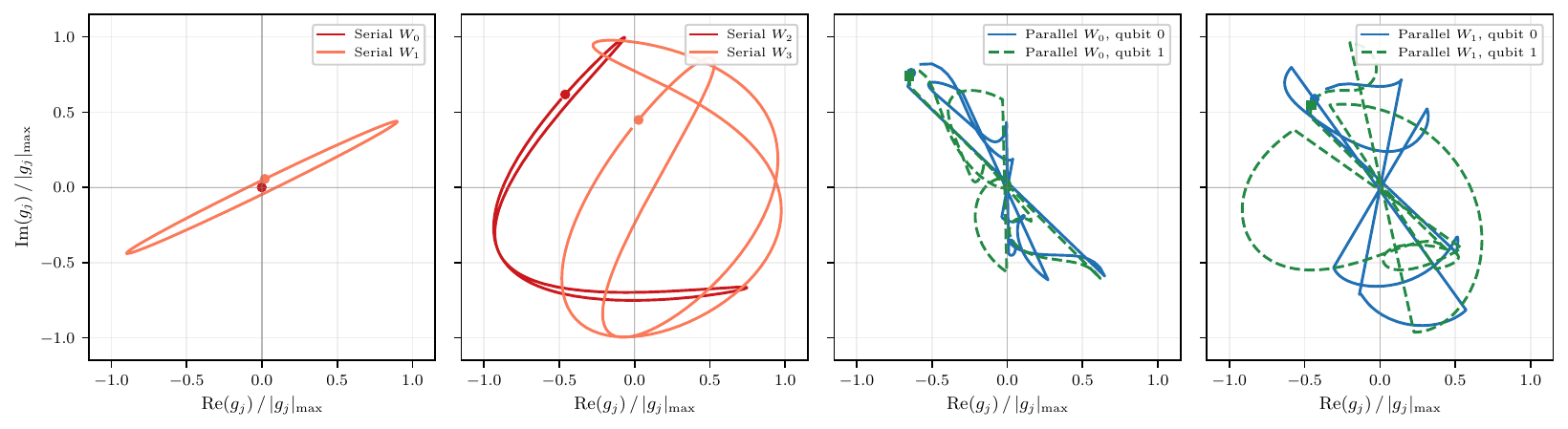}
  \caption{
  \textbf{Phase-locking in serial circuits vs.\ independent
  phase trajectories in parallel circuits.}
  Each curve shows $g_j(x) = \langle\psi^\perp|\partial_j\psi
  \rangle \in \mathbb{C}$, normalized by its maximum magnitude,
  as $x$ sweeps $[0,2\pi]$ and plotted in the complex plane.
  \textbf{Left two panels (serial, $N=1$, $L=4$):} parameters
  from the first four ansatz blocks $W_0,\ldots,W_3$ trace
  curves lying entirely on the real axis of $\mathbb{C}$,
  confirming that $g_j = h_j \in \mathbb{R}$ exactly
  (Lemma~\ref{lemma:serial_coupling}, Appendix~\ref{app:lemma}):
  the global phase factor $G = e^{i\varphi} = 1$ identically,
  so $g_j/g_k \in \mathbb{R}$ at every $x$
  ($\mathrm{Im}(g_j/g_k) = 0.000$ exactly).
  Curves are related by real scaling factors $h_j/h_k$,
  the direct signature of phase-locking.
  \textbf{Right two panels (parallel, $N=2$, $L=2$):}
  parameters on qubit~0 (solid) and qubit~1 (dashed) from
  the same ansatz block trace curves with genuinely different
  shapes that leave the real axis, reflecting independent
  local phases $\varphi^{(0)} \neq \varphi^{(1)}$
  ($|\mathrm{Im}(g_j/g_k)| = 0.19$ on average).
  This holds consistently for both blocks $W_0$ and $W_1$,
  confirming it is a structural property of the parallel
  architecture rather than a parameter-specific coincidence.
  }
  \label{fig:phase_locking}
\end{figure}

\section{FM Layers vs.\ Trainable Blocks: A Structural Comparison}
\label{sec:op}
Section~\ref{sec:theory} establishes the rank ceiling $2NL+1$
as the fundamental structural constraint on trainability:
serial circuits are confined to this ceiling by phase-locking,
while parallel circuits maintain full column rank up to it.
This ceiling is not fixed --- it is a property of the encoding
structure, not the ansatz, and can be raised by adding FM layers.
This observation directly motivates the central question of this
section: given a parallel circuit where gradient starvation is
not yet the bottleneck, which route to more parameters is more
effective --- adding \emph{feature map (FM) layers}, which increase $P$
while simultaneously expanding the Fourier spectrum and raising
the rank ceiling, or adding \emph{trainable blocks}, which
increase $P$ without changing either?
This section shows these two routes have fundamentally different
effects on trainability, with the FM route achieving
$R^2 \geq 0.95$ with $1.6$--$2.2\times$ fewer parameters
across all tested architectures.
Full experimental details and degree justification are given in
Appendices~\ref{app:experiments} and~\ref{app:degree_justification}.

\subsection{Two Routes to More Parameters}
\label{subsec:two_routes}

Starting from a circuit with $L$ FM layers and one trainable
block layer per encoding block, there are two natural ways to
add parameters while keeping the architecture parallel:

\begin{enumerate}
  \item \textbf{FM route}: increase the number of FM layers $L$,
    keeping the trainable block depth $\mathrm{tbl}=1$ fixed.
    Each extra FM layer adds $N \cdot 3$ parameters and
    simultaneously expands $|\Omega|$, raising the Jacobian rank
    ceiling and making more Fourier directions accessible to the
    optimizer.

  \item \textbf{Trainable blocks route}: increase the number of
    trainable block layers $\mathrm{tbl}$, keeping $L = L_{\min}$
    fixed.
    Each extra block adds $(L_{\min} + 1) \cdot N \cdot 3$ parameters
    without changing the Fourier spectrum or redundancy structure.
    The Jacobian rank ceiling remains fixed.
\end{enumerate}

As established in Corollary~\ref{cor:parallel}, parallel circuits
have a rank ceiling of $|\Omega| = 2E+1$ beyond which additional
parameters create null directions in $J$.
Adding FM layers raises this ceiling; adding trainable blocks increases $P$ without raising the
ceiling, relying on the classical interpolation mechanism
($P \geq n_{\mathrm{train}}$) with no quantum-specific benefit.
Both routes assume $L \geq L_{\min} = \lceil\mathrm{deg}/N\rceil$
FM layers are already present; below $L_{\min}$ the target frequencies
lie outside the circuit's function class and neither route can improve
training. In all experiments we fix $L = L_{\min}$ for the trainable
blocks route and $\mathrm{tbl}=1$ for the FM route, isolating the
effect of each.

\subsection{The FM Route: Rank Growth}
\label{subsec:fm_route}
Figure~\ref{fig:FM_vs_tbl} shows both routes for
$N \in \{1, 6\}$, representing contrasting efficiency regimes:
$N=1$ (strong gradient starvation, large target degree relative
to encoding budget) and $N=6$ (mild gradient starvation, target
degree well within the circuit's accessible spectrum at
$L_{\min}$).
As $L$ increases from $L_{\min}$, the normalized Jacobian QFIM
eigenvalue spectra $\lambda_i / \lambda_1$ shift visibly to the
right.
We characterize this by the \emph{spectral knee} --- the rank
index beyond which $\lambda_i/\lambda_1$ drops sharply --- a
threshold-free measure of the number of genuinely exploitable
gradient directions.
Along the FM route, the knee generally shifts rightward
with each extra layer: at $L=12$ the spectrum drops sharply
at rank $25$; by $L=22$ this extends to rank $33$.
Along the trainable blocks route, the spectral knee is
frozen at the theoretical rank ceiling $2NL_{\min}+1$ for
all tested architectures: $N=1$ freezes at rank $24$--$25$
(ceiling $25$), $N=2$ at rank $\approx 41$ (ceiling $41$),
$N=4$ at rank $\approx 57$ (ceiling $57$), and $N=6$ at
rank $12$--$13$ (ceiling $13$)
(Appendix~\ref{app:n2_n4_fm_tbl},
Figure~\ref{fig:FM_vs_tbl} bottom right).
This confirms that adding trainable blocks provides access
to no new Fourier directions regardless of architecture:
the rank ceiling imposed by $L_{\min}$ is tight, and the
tbl route can only improve coverage of directions already
within that ceiling --- a purely classical interpolation
mechanism with no quantum-specific benefit.

\begin{figure}[t]
  \centering
  \includegraphics[width=\linewidth]{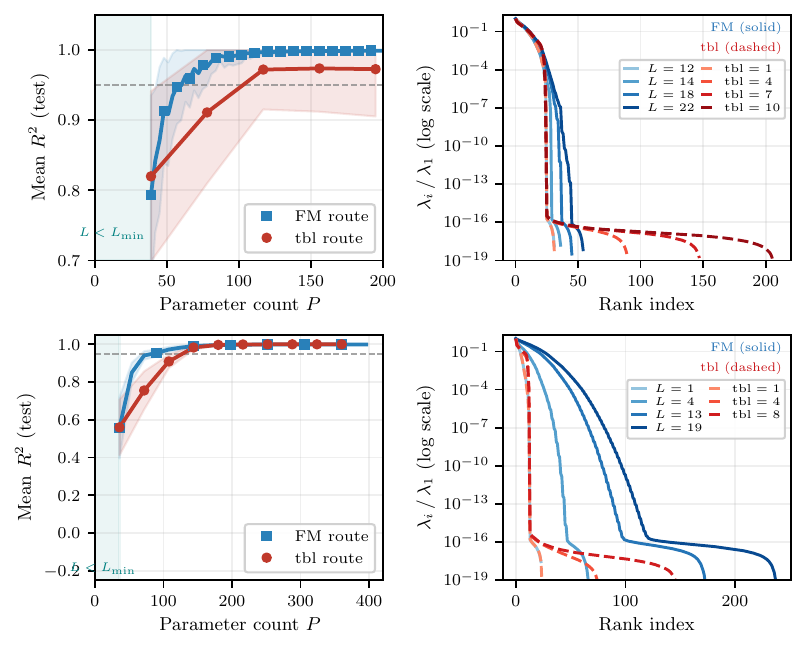}
  \caption{
    \textbf{FM route vs.\ trainable blocks route for $N=1$
(gradient starvation) and $N=6$ (high-qubit regime).}
    \textbf{Left panels:} mean $R^2$ (test) vs.\ parameter
    count $P$ for the FM route (blue, solid) and trainable
    blocks route (red, dashed).
    Teal shading: expressivity gap ($L < L_{\min}$).
    Dashed gray line: $R^2 = 0.95$ threshold.
    \textbf{Right panels:} normalized Jacobian QFIM eigenvalue
    spectra $\lambda_i/\lambda_1$ at selected $L$ values (FM,
    solid blue) and $\mathrm{tbl}$ values (trainable blocks,
    dashed red).
    \emph{Top row ($N=1$):} the spectral knee shifts rightward
    with each extra FM layer while remaining frozen
    at the theoretical rank ceiling $2L_{\min}+1 = 25$, despite $R^2$ improving along
    both routes.
    \emph{Bottom row ($N=6$):} both routes reach $R^2 \geq 0.95$.
    The spectral knee shifts rightward along the FM route and
    remains frozen along the tbl route, consistent with the
    structural prediction of Section~\ref{subsec:fm_route}.
    Parameter efficiency ratios for all architectures are given
    in Table~\ref{tab:efficiency}.
  }
  \label{fig:FM_vs_tbl}
\end{figure}

\paragraph{Spectral knee as a trainability diagnostic.}
The spectral knee --- the rank index beyond which
$\lambda_i/\lambda_1$ drops sharply --- distinguishes the two
routes structurally.
Along the FM route, the knee grows monotonically with each extra
layer because the ceiling $2NL+1$ itself expands: each FM layer
makes genuinely new Fourier directions accessible to the optimizer.
Along the trainable blocks route, the ceiling is fixed at
$2NL_{\min}+1$; the knee is frozen at this ceiling throughout
the tbl sweep, confirming that no new Fourier directions become
accessible via the classical interpolation mechanism.
This distinction is quantified in
Appendix~\ref{app:spectral_knee} (Figure~\ref{fig:spectral_knee}).

\subsection{The Trainable Blocks Route: Interpolation Threshold}
\label{subsec:tbl_route}

Despite the spectral knee being frozen at the fixed
$L_{\min}$ ceiling,, $R^2$ still improves substantially
with $\mathrm{tbl}$.
The governing mechanism is the classical interpolation threshold:
once the parameter count $P$ exceeds the training set size
$n_{\mathrm{train}}$, the system $f(x_i;\boldsymbol{\theta}) = y_i$
becomes overdetermined and gradient descent finds an interpolating
solution~\citep{belkin2019reconciling,bartlett2020benign}.
This is a purely classical phenomenon with no dependence on
quantum structure.
The spectral knee of $\hat{\mathcal{F}}$ remains frozen
at the $L_{\min}$ ceiling throughout the $\mathrm{tbl}$
sweep, confirming that no new Fourier directions become
accessible as trainable blocks are added.
For $N=6$, deg$=6$, both routes reach $R^2 \geq 0.95$:
the FM route at $L^*=4$ ($P=90$) and the trainable blocks
route at tbl$^*=4$ ($P=144$), giving a $1.6\times$ efficiency
ratio.
The lower ratio relative to $N=1,2,4$ reflects that the
circuit is already expressive enough at $L_{\min}=1$ for a
degree-6 target, so extra FM layers provide less marginal
benefit than at higher target degrees.

Table~\ref{tab:interp} (Appendix~\ref{app:interp_table}) shows
that $R^2 \geq 0.95$ is first achieved within 1--3 $\mathrm{tbl}$
steps of the interpolation threshold $P \geq n_{\mathrm{train}} = 200$
across all architectures that reach saturation, consistent with
this mechanism.

\subsection{Efficiency Comparison}
\label{subsec:efficiency}

Table~\ref{tab:efficiency} (Appendix~\ref{app:efficiency_table}) 
quantifies the parameter
efficiency of both routes for architectures where both
reach saturation. The FM route requires $1.6$--$2.2\times$
fewer parameters across all four architectures: each new
FM layer raises the rank ceiling, keeping a growing fraction
of parameters exploitable, while the tbl route keeps the
ceiling fixed and relies on classical interpolation.
This advantage spans $N=1$ to $N=6$ and target degrees
from $6$ to $28$, confirming it reflects a structural
property of the FM route rather than an architecture-dependent
accident.
The FM route is the only route that engages a quantum-specific
mechanism: it simultaneously expands the accessible frequency spectrum
and shifts the spectral knee rightward in a way that has no classical
analog.

\subsection{Validation on Real-World Data}
\label{subsec:realworld}
The structural predictions of Sections~\ref{sec:theory}
and~\ref{sec:op} are validated on the Nottingham temperature
dataset --- 240 monthly mean air temperatures recorded at
Nottingham Castle from 1920 to 1939 --- in
Appendix~\ref{app:realworld}, using two experiments that
probe different regimes.

With encoding budget $E = N \times L_{\min} = 12$
(Figure~\ref{fig:realworld_deg12}), the model spectrum is
insufficient at $L_{\min}$ for all architectures: the
Nottingham spectrum extends beyond $L_{\min}$ coverage,
so FM layers provide a dual benefit --- expanding Jacobian
rank coverage \emph{and} expressivity --- while the trainable
blocks route fails categorically ($R^2 < 0$ throughout) for
all four architectures, unable to expand either.

With encoding budget $E = N \times L_{\min} = 24$
(Figure~\ref{fig:realworld_deg24}), expressivity is
sufficient at $L_{\min}$ and the experiment isolates
parameter efficiency.
The FM route achieves $1.9$--$2.7\times$ fewer parameters
than the tbl route for $N \in \{1,2,4\}$, consistent with
the synthetic results of Table~\ref{tab:efficiency}.
For $N=6$, the result inverts: the tbl route reaches
$R^2 \geq 0.95$ while the FM route plateaus at
$R^2 \approx 0.83$, consistent with the $N=6$ circuit
entering the Larocca underparameterization regime
$P \approx R_{\max}=126$~\citep{larocca2023theory} early
in the FM sweep, where adding encoding layers increases
the parameter demand faster than the added parameters
satisfy it.

\section{Conclusion}
\label{sec:conclusion}

The trainability gap between serial and parallel QNN architectures
at fixed encoding budget $E = NL$ has a precise mechanistic
explanation.
The $\mathbb{CP}^1$ geometry of a single-qubit state space forces
all $P$ gradient directions to share a common global phase
(Lemma~\ref{lemma:serial_coupling}, Figure~\ref{fig:phase_locking}),
confining the column space of the coefficient matching Jacobian $J$
to a subspace of dimension $2L+1$: as $P$ grows at fixed $L$, the
fraction of parameters receiving identically zero gradient signal
approaches 1 --- \emph{structural gradient starvation}.
Parallel architectures defer this via independent phase trajectories
across qubits, maintaining $\sigma_{\min}(J^{(\mathrm{par})}) > 0$
generically for $P \leq 2E+1$, so every parameter contributes to
at least one exploitable Fourier direction up to the shared onset
threshold.

\paragraph{Add FM layers, not trainable blocks.}
The structural analysis directly motivates the practical
recommendation.
FM layers simultaneously expand the accessible Fourier spectrum
and shift the spectral knee of the Jacobian QFIM rightward ---
a quantum-specific mechanism with no classical analog.
Trainable blocks increase $P$ without changing either the
spectrum or the spectral knee, improving training solely via
the classical interpolation threshold ($P \geq n_{\mathrm{train}}$).
The $1.6$--$2.2\times$ parameter efficiency advantage of the FM route
is a direct consequence of this structural difference: FM
parameters are added below the rank ceiling, while trainable
block parameters accumulate in $\ker J$.

\paragraph{Limitations and future work.}
Theoretical results are proved for the architecture extremes
($N=1$ serial, product ansatz parallel); extension to hybrid
architectures and general entangling ans\"{a}tze remains open.
The efficiency advantage of the FM route is characterized for 
circuits operating below the geometric threshold 
$R_{\max} = 2^{N+1}-2$ of~\citet{larocca2023theory}; 
in this underparameterized regime, trainable blocks may offer 
competitive parameter efficiency by better exploiting existing 
state-space directions before the accessible spectrum is fully 
saturated, and a more complete characterization of the FM 
vs.\ trainable blocks trade-off in this regime remains open.
Real-world validation on the Nottingham temperature dataset
confirms the structural predictions and identifies the
boundary of the FM efficiency advantage.
With encoding budget $E = N \times L_{\min} = 12$, the
trainable blocks route fails categorically for all
architectures while the FM route succeeds by simultaneously
expanding expressivity and Jacobian rank coverage.
With $E = N \times L_{\min} = 24$, the FM route achieves
$1.9$--$2.7\times$ fewer parameters than the tbl route for
$N \in \{1,2,4\}$, consistent with the synthetic results.
The exception is $N=6$, where the FM advantage inverts:
the tbl route reaches $R^2 \geq 0.95$ while the FM route
plateaus, consistent with the circuit entering the Larocca
underparameterization regime early in the FM sweep.
This identifies a natural boundary condition for the FM
efficiency advantage and motivates a more complete
characterization of the FM vs.\ trainable blocks trade-off
in the Larocca-limited regime as future work.

\section*{Acknowledgments}
This work has been supported by the LMU Sustainability Fund (EfOiE),
the BMFTR (QuCUN, QuaRDS, CAQAO), the Munich Quantum Valley (K5, K7),
and the Bavarian StMWi (6GQT).
The sole responsibility for the report's contents lies with the authors.

\bibliography{neurips_2026}
\bibliographystyle{neurips_2026}


\clearpage
\appendix
\renewcommand{\thesection}{\Alph{section}}

\section{Proof of Proposition~\ref{prop:bilinear}}
\label{app:bilinear}

\begin{proposition}[Bilinear factorization for general parallel
  architectures]
\label{prop:bilinear}
For a parallel $N$-qubit architecture with $L=1$ encoding layer,
the Fourier coefficients of the circuit output are:
\begin{equation}
  c_\omega^{(\mathrm{par})}
  = \sum_{\substack{j,k \in \{0,\ldots,2^N-1\} \\
    \sum_n(k_n - j_n) = \omega}}
    \bar{v}_k\,\widetilde{M}_{kj}\,v_j,
\end{equation}
where $\mathbf{v} = W^{(1)}|0\rangle^{\otimes N}$ and
$\widetilde{M} = (W^{(2)})^\dagger M\, W^{(2)}$.
Every term is \emph{bilinear}: $\bar{v}_k v_j$ depends only on
$W^{(1)}$ and $\widetilde{M}_{kj}$ only on $W^{(2)}$, with no
cross-coupling between the two ansatz layers.
\end{proposition}

\begin{proof}
\textbf{Step 1: Diagonal encoding and circuit output.}
Since any Pauli operator $\sigma$ is unitarily equivalent to
$\sigma_z$, the FM layer $S(x) = e^{ixH}$ with
$H = \frac{1}{2}\sigma$ is diagonal in an appropriate
single-qubit basis, with eigenvalues $\pm\frac{1}{2}$.
Working in this basis, $S(x)|k\rangle = e^{i\phi_k(x)}|k\rangle$
where
\begin{equation}
  \phi_k(x) = \frac{x}{2}\sum_{n=1}^N(1-2k_n),
\end{equation}
with $k_n \in \{0,1\}$ the $n$-th bit of $k$.
Setting $\mathbf{v} = W^{(1)}|0\rangle^{\otimes N}$ and
$\widetilde{M} = (W^{(2)})^\dagger M\, W^{(2)}$, the circuit
output expands as:
\begin{equation}
  f(x;\boldsymbol{\theta})
  = \langle\mathbf{v}|\,e^{-ixH}\widetilde{M}\,e^{ixH}|\mathbf{v}\rangle
  = \sum_{j,k=0}^{2^N-1}
    \bar{v}_k\,\widetilde{M}_{kj}\,v_j\,
    e^{i(\phi_k(x)-\phi_j(x))}.
  \label{eq:app_bilinear_expand}
\end{equation}

\textbf{Step 2: Frequency decomposition.}
The Fourier frequency of term $(j,k)$ is:
\begin{equation}
  \omega_{jk}
  = \phi_k(x)/x - \phi_j(x)/x
  = \sum_{n=1}^N(k_n - j_n)
  \in \{-N,\ldots,N\}.
\end{equation}
Collecting all terms with the same frequency $\omega$ gives:
\begin{equation}
  c_\omega
  = \sum_{\substack{j,k \in \{0,\ldots,2^N-1\}\\
    \sum_n(k_n-j_n)=\omega}}
    \bar{v}_k\,\widetilde{M}_{kj}\,v_j.
  \label{eq:app_c_omega}
\end{equation}
This sum includes all pairs $(j,k)$ with the appropriate net
qubit difference, including non-complementary pairs where
$j_n = k_n$ at some positions.

\textbf{Step 3: Bilinearity.}
Every term in~\eqref{eq:app_c_omega} factorizes as:
\begin{equation}
  \bar{v}_k\,\widetilde{M}_{kj}\,v_j
  = \underbrace{\bar{v}_k\,v_j}_{W^{(1)}\text{-factor}}
  \cdot
  \underbrace{\widetilde{M}_{kj}}_{W^{(2)}\text{-factor}}.
\end{equation}
The entries $v_j, v_k$ of $\mathbf{v} = W^{(1)}|0\rangle^{\otimes N}$
depend only on $W^{(1)}$ and are independent of $W^{(2)}$.
The entries $\widetilde{M}_{kj}$ of
$(W^{(2)})^\dagger M\, W^{(2)}$ depend only on $W^{(2)}$
and are independent of $W^{(1)}$.
No term couples $W^{(1)}$ and $W^{(2)}$ multiplicatively,
establishing the bilinear separation for all frequencies
$\omega \in \{-N,\ldots,N\}$.
The $N=2$ case recovers equation~\eqref{eq:par_c2} directly.
\end{proof}

\section{Full Proof of Proposition~\ref{prop:rank}}
\label{app:prop1}

\begin{proof}
We work at fixed $(x,\boldsymbol{\theta})$, writing
$|\psi\rangle=|\psi(x,\boldsymbol{\theta})\rangle$.

\textbf{Step 1: Tangent space of $\mathbb{CP}^1$.}
For $N=1$, normalized pure states form $\mathbb{CP}^1\cong S^2$,
a smooth real-2-dimensional manifold~\citep{bengtsson2006geometry}.
At any $|\psi\rangle$, the tangent space
$T_{|\psi\rangle}\mathbb{CP}^1$ is spanned over $\mathbb{C}$ by
the unique (up to phase) unit vector $|\psi^\perp\rangle$ with
$\langle\psi|\psi^\perp\rangle=0$.
Any infinitesimal variation satisfies:
\begin{equation}
  |\partial_j\psi\rangle
  = g_j|\psi^\perp\rangle + \lambda_j|\psi\rangle,
  \qquad g_j\in\mathbb{C},\;\lambda_j\in i\mathbb{R},
\end{equation}
where $\lambda_j\in i\mathbb{R}$ follows from
$\langle\psi|\psi\rangle=1$.

\textbf{Step 2: The gradient of $f$.}
$f=\langle\psi|M|\psi\rangle$ for Hermitian $M$.
Differentiating and using
$\bar\lambda_j\langle\psi|M|\psi\rangle\in i\mathbb{R}$
(vanishes under $\mathrm{Re}[\cdot]$):
\begin{equation}
  \frac{\partial f}{\partial\theta_j}
  = 2\,\mathrm{Re}[\bar g_j\langle\psi^\perp|M|\psi\rangle].
\end{equation}

\textbf{Step 3: Shared scalar $\mu$.}
Define $\mu:=\langle\psi^\perp|M|\psi\rangle\in\mathbb{C}$,
which is the same for every $j$.
Writing $\mu=a+ib$ and $g_j=u_j+iw_j$:
\begin{equation}
  \frac{\partial f}{\partial\theta_j}
  = 2(au_j + bw_j)
  = 2\bigl[a\,\mathrm{Re}(g_j) + b\,\mathrm{Im}(g_j)\bigr].
\end{equation}
Hence
$\nabla_{\boldsymbol{\theta}} f = 2a\,\mathrm{Re}(g) +
2b\,\mathrm{Im}(g)
\in\mathrm{span}\{\mathrm{Re}(g),\mathrm{Im}(g)\}$,
which has dimension $\leq 2$.

\textbf{Step 4: Rank of the QFIM.}
\begin{equation}
  \mathcal{F}_{jk}
  = 4\,\mathrm{Re}[\bar g_j g_k]
  = 4\bigl(\mathrm{Re}(g)_j\mathrm{Re}(g)_k
    + \mathrm{Im}(g)_j\mathrm{Im}(g)_k\bigr),
\end{equation}
a sum of two rank-one matrices~\citep{stokes2020quantum},
so $\mathrm{rank}(\mathcal{F})\leq 2$.

\textbf{Tightness.}
For generic $M$ and $|\psi\rangle$, $\mu\neq 0$ and
$g\not\in\mathbb{R}^P$, so $\mathrm{Re}(g)$ and $\mathrm{Im}(g)$
are linearly independent and the bound is tight.
\end{proof}

\section{Full Proof of Lemma~\ref{lemma:serial_coupling}}
\label{app:lemma}

\begin{proof}
Fix $(x,\boldsymbol{\theta})$.
For each $j\in\{1,\ldots,L\}$ define:
\begin{align}
  D_j &:= \prod_{l=L}^{j+1}V_l(\boldsymbol{\theta}_l)
    S_l(x)\in\mathrm{SU}(2),\\
  |\phi_j\rangle &:=
    V_j(\boldsymbol{\theta}_j)\,S_j(x)
    \prod_{l=j-1}^{1}V_l(\boldsymbol{\theta}_l)\,S_l(x)
    |\psi_0\rangle
    \in\mathbb{C}^2,
\end{align}
so that $|\psi\rangle=D_j|\phi_j\rangle$ and
$|\partial_j\psi\rangle=D_j(\partial_j V_j)V_j^{-1}|\phi_j\rangle$.

\textbf{Step 1: Gradient coefficient.}
Since $|\psi\rangle = D_j|\phi_j\rangle$ and $D_j \in \mathrm{SU}(2)$
is unitary, the orthogonal complement transforms as
$|\psi^\perp\rangle = D_j|\hat\phi_j^\perp\rangle$, giving
$D_j^\dagger|\psi^\perp\rangle = |\hat\phi_j^\perp\rangle$.
Differentiating $|\psi\rangle = D_j|\phi_j\rangle$ with respect
to $\theta_j$ and using the fact that $D_j$ does not depend on
$\theta_j$ (it is the sub-circuit \emph{above} layer $j$):
\begin{equation}
  |\partial_j\psi\rangle = D_j\,(\partial_j V_j)\,V_j^{-1}|\phi_j\rangle.
\end{equation}
The gradient coefficient is therefore:
\begin{equation}
  g_j = \langle\psi^\perp|\partial_j\psi\rangle
  = \langle\psi^\perp|\,D_j\,(\partial_j V_j)V_j^{-1}|\phi_j\rangle
  = \langle D_j^\dagger\psi^\perp|
    (\partial_j V_j)V_j^{-1}|\phi_j\rangle,
\end{equation}
where the last equality uses unitarity of $D_j$.
The key observation is that the right-hand factor
$(\partial_j V_j)V_j^{-1}|\phi_j\rangle$ depends only on
the sub-circuit \emph{up to and including} layer $j$, while
$D_j^\dagger|\psi^\perp\rangle$ depends only on the sub-circuit
\emph{above} layer $j$. Step 2 shows that the latter is
independent of $j$ up to a global phase.

\textbf{Step 2: Canonical gauge and phase independence.}
To eliminate the $\mathrm{U}(1)$ ambiguity in the definition of
$|\psi^\perp\rangle$, we fix a canonical gauge via the
anti-unitary time-reversal operator $T = i\sigma_y K$,
setting $|\psi^\perp\rangle := T|\psi\rangle$ and
$|\hat\phi_j^\perp\rangle := T|\hat\phi_j\rangle$ for all $j$.
This choice is $\mathrm{U}(1)$-equivariant: if $|\psi\rangle
\mapsto e^{i\alpha}|\psi\rangle$ then $T|\psi\rangle \mapsto
e^{-i\alpha}T|\psi\rangle$, so relative phases between
$|\psi^\perp\rangle$ and $|\hat\phi_j^\perp\rangle$ are
preserved independently of $j$.
For any $D \in \mathrm{SU}(2)$ with $D =
\bigl(\begin{smallmatrix}a & -\bar{b} \\ b & \bar{a}
\end{smallmatrix}\bigr)$, direct computation gives
\begin{equation}
  (i\sigma_y)\,D^* = D\,(i\sigma_y),
\end{equation}
which implies $T D = D T$ and hence $T D_j^\dagger = D_j^\dagger T$.
Using this together with $|\psi^\perp\rangle = T|\psi\rangle$:
\begin{equation}
  D_j^\dagger |\psi^\perp\rangle
  = D_j^\dagger T|\psi\rangle
  = T D_j^\dagger |\psi\rangle
  = T|\hat\phi_j\rangle
  = |\hat\phi_j^\perp\rangle,
\end{equation}
where the third equality uses $D_j^\dagger|\psi\rangle = |\hat\phi_j\rangle$
by construction.
The phase factor is therefore:
\begin{equation}
  e^{i\varphi_j}
  = \langle T\hat\phi_j \,|\, D_j^\dagger T\psi \rangle
  = \langle T\hat\phi_j \,|\, T D_j^\dagger \psi \rangle
  = \overline{\langle \hat\phi_j \,|\, D_j^\dagger \psi \rangle}
  = \overline{\langle \hat\phi_j \,|\, \hat\phi_j \rangle}
  = 1,
\end{equation}
where the third equality uses the anti-unitarity of $T$:
$\langle Ta \,|\, Tb \rangle = \overline{\langle a \,|\, b \rangle}$.
Hence $G = e^{i\varphi} = 1$ independently of $j$, $x$, and
$\boldsymbol{\theta}$: the global phase factor is identically unity
and each $g_j = h_j$ is real-valued.

\textbf{Step 3: Factorization.}
Since $D_j^\dagger|\psi^\perp\rangle = |\hat\phi_j^\perp\rangle$
(Step 2), substituting back:
\begin{equation}
  g_j(x) =
  \underbrace{
    \langle\hat\phi_j^\perp|
    (\partial_j V_j)V_j^{-1}|\phi_j\rangle
  }_{h_j(x) \,\in\, \mathbb{R}},
\end{equation}
establishing $g_j = h_j$ with $h_j$ real-valued.
The per-input gradient rank bound $\leq 2$ then follows from
Proposition~\ref{prop:rank}.

\textbf{Step 4: Phase-locking and rank of $J$.}
Since $g_j = h_j \in \mathbb{R}$ for all $j$, the subspace
$\mathcal{V}(x,\boldsymbol{\theta}) =
\mathrm{span}\{\mathrm{Re}(g),\mathrm{Im}(g)\}$
collapses: $\mathrm{Im}(g_j) = 0$ identically, so
$g_j/g_k \in \mathbb{R}$ at every $x$ --- the phase-locking
of Lemma~\ref{lemma:serial_coupling} is exact, not merely
approximate.
Each $h_j$ is determined by $j$ encoding gates $S_1,\ldots,S_j$
and is therefore a trigonometric polynomial of degree
$\leq j \leq L$~\citep{schuld2021effect}.
All $\{h_j\}_{j=1}^P$ therefore lie in
$\mathcal{T}_L = \mathrm{span}\{e^{ikx}:|k|\leq L\}$,
a real vector space of dimension $\leq 2L+1$.
Since $\varphi$ is common to all columns, the column space of $J$
has dimension $\leq 2L+1$ independently of $P$.
\end{proof}

\section{Full Proof of Proposition~\ref{prop:ill_cond}}
\label{app:prop2}

\begin{proof}
\textbf{Step 1: Jacobian entry.}
From Proposition~\ref{prop:rank} and
Lemma~\ref{lemma:serial_coupling}, substituting
$g_j = h_j e^{i\varphi}$:
\begin{equation}
  J_{\omega,j}
  = \int_0^{2\pi}\alpha(x)\,\mathrm{Re}(h_j)\,e^{-i\omega x}\,dx
  + \int_0^{2\pi}\beta(x)\,\mathrm{Im}(h_j)\,e^{-i\omega x}\,dx,
  \label{eq:J_AB}
\end{equation}
where $\alpha(x) := 2(a\cos\varphi + b\sin\varphi)$ and
$\beta(x) := 2(b\cos\varphi - a\sin\varphi)$ are independent
of $j$, with $a, b := \mathrm{Re}(\mu), \mathrm{Im}(\mu)$
from Proposition~\ref{prop:rank}.

\textbf{Step 2: Degree bound on $h_j$.}
Each encoding gate $S_l(x) = e^{ix\sigma/2}$ contributes one
factor of $e^{\pm ix/2}$, so $h_j \in \mathcal{T}_L$ with
$h_j(x) = \sum_{|k|\leq j}\hat h_j^{(k)}e^{ikx}$
and $j\leq L$~\citep{schuld2021effect}.
All $P$ columns of $J$ therefore lie in $\mathcal{T}_L$.

\textbf{Step 3: Rank bound and tightness (i).}
The bound $\mathrm{rank}(J) \leq 2L+1$ follows from
$J \in \mathbb{R}^{(2L+1)\times P}$.
Tightness holds because the map $h_j \mapsto J_{\cdot,j}$
is a Fourier projection into $\mathcal{T}_L$, and for generic
$\boldsymbol{\theta}$ the $2L+1$ rows of $J$ corresponding
to distinct frequencies $\omega \in \Omega$ are linearly
independent --- established by the real-analyticity argument
applied to $\det(JJ^\top)$, which is not identically zero
since the circuit is expressive~\citep{schuld2021effect}.
Hence $\mathrm{rank}(J) = 2L+1$ generically, proving~(i).

\textbf{Step 4: Growth of the null space (ii).}
For $P > 2L+1$, the rank-nullity theorem applied to Step~3 gives:
\begin{equation}
  \dim(\ker J) = P - \mathrm{rank}(J) \geq P - (2L+1),
\end{equation}
which grows without bound as $P \to \infty$ at fixed $L$.
Every parameter $\boldsymbol{\theta}_j \in \ker J$ satisfies
$J_{\omega,j} = 0$ for all $\omega \in \Omega$, meaning it
produces no change in any Fourier coefficient and receives
identically zero gradient signal.
The fraction of such parameters,
\begin{equation}
  \frac{\dim(\ker J)}{P} \geq \frac{P - (2L+1)}{P} \to 1
  \quad \text{as } P \to \infty,
\end{equation}
grows monotonically toward 1, so the overwhelming majority of
parameters become structurally decoupled from the loss as $P$
increases at fixed $L$.
This is independent of the condition number of the exploitable
subspace $\mathrm{im}(J^\top)$, which has dimension at most
$2L+1$ throughout.
\end{proof}

\begin{remark}[Encoding-dependence vs.\ architecture-dependence]
\label{rem:encoding_vs_arch}
Spectral bias is \emph{encoding-dependent}: redundancy numbers
$\binom{2E}{E-\omega}$ crowd out high-frequency gradients under
unary encoding but are reduced under richer schemes
(e.g.\ Chebyshev or non-uniform encodings).
Structural gradient starvation is \emph{architecture-dependent}:
the $\mathbb{CP}^1$ geometry bounds per-input gradient rank at 2
regardless of encoding, and phase-locking caps the Jacobian rank
at $2L_{\mathrm{eff}}+1$ where $L_{\mathrm{eff}}$ scales with
encoding richness.
Richer encodings therefore shift \emph{when} structural gradient
starvation sets in, but not \emph{whether} it does.
\end{remark}

\section{Full Proof of Corollary~\ref{cor:parallel}}
\label{app:cor}

\begin{proof}
\textbf{Step 1: Structure of $J^{(\mathrm{par})}$ and the
phase-locking argument.}
For a product ansatz $U = \bigotimes_n U^{(n)}$ with a tensor
product observable, the global Fourier coefficients are
convolutions of local coefficients:
\begin{equation}
  c_\omega = \sum_{\omega_1+\cdots+\omega_N=\omega}
  \prod_{n=1}^N c^{(n)}_{\omega_n}.
\end{equation}
Differentiating with respect to $\theta^{(m)}_j$ yields
\begin{equation}
  \frac{\partial c_\omega}{\partial \theta^{(m)}_j}
  = \sum_{\omega_1+\cdots+\omega_N=\omega}
  \frac{\partial c^{(m)}_{\omega_m}}{\partial \theta^{(m)}_j}
  \cdot \prod_{n \neq m} c^{(n)}_{\omega_n},
\end{equation}
which depends on the coefficients of all qubits $n \neq m$
via $\prod_{n\neq m} c^{(n)}_{\omega_n}$, so
$J^{(\mathrm{par})}$ is generically dense.

The parallel advantage is instead established via the
phase-locking argument of Lemma~\ref{lemma:serial_coupling}.
For $N=1$, the phase $\varphi(x,\boldsymbol{\theta})$ in
$g_j = h_j e^{i\varphi}$ is determined by the unique state
$|\psi\rangle \in \mathbb{CP}^1$ and is identical for all
$P$ parameters, confining
$\mathrm{rank}(J^{(\mathrm{ser})}) \leq 2L+1$ regardless of $P$.
For $N > 1$, each qubit $n$ has its own downstream operator
$D_j^{(n)}$ and local reduced state $|\psi^{(n)}\rangle$,
so the local phase $\varphi^{(n)}$ is qubit-dependent.
Parameters on qubit $m$ are modulated by $e^{i\varphi^{(m)}}$
while parameters on qubit $n \neq m$ are modulated by
$e^{i\varphi^{(n)}}$, with $\varphi^{(m)} \neq \varphi^{(n)}$
generically.
The global coherence condition of
Lemma~\ref{lemma:serial_coupling} is therefore broken: columns
of $J^{(\mathrm{par})}$ corresponding to different qubits are
modulated by independent phases and are not confined to a single
$\mathcal{T}_L$-subspace.

\textbf{Step 2: Onset threshold at fixed $E$.}
The shared frequency spectrum
$\Omega = \{-E, \ldots, E\}$
(Theorem~\ref{thm:spectral_invariance}) implies
$|\Omega| = 2E+1$ rows in $J^{(\mathrm{par})}$ regardless
of $N$.
By rank-nullity:
\begin{equation}
  \dim(\ker J^{(\mathrm{par})}) \geq P - (2E+1) > 0
  \quad \text{for } P > 2E+1,
\end{equation}
so null directions appear at threshold $P \approx 2E+1$,
consistent with Proposition~\ref{prop:ill_cond}.

\textbf{Step 3: Full column rank for $P \leq 2E+1$.}
Since the column space of $J^{(\mathrm{par})}$ spans more than
$2L+1$ dimensions generically --- by absence of phase-locking
--- $\sigma_{\min}(J^{(\mathrm{par})}) > 0$ for $P \leq 2E+1$.
Formally, the map
$\Phi: \mathbb{R}^P \to \mathbb{R}^{2E+1}$
from parameters to real Fourier representatives
(Definition~\ref{def:jacobian}) is real-analytic, and the
circuit is expressive for $P \leq 2E+1$, so
$q(\boldsymbol{\theta}) := \det(J^{(\mathrm{par})\top}
J^{(\mathrm{par})})$ is not identically zero.
By the real-analyticity argument~\citep[Prop.~2.2.7]{bochnak1998real},
$q > 0$ for almost all $\boldsymbol{\theta}$, giving
$\sigma_{\min}(J^{(\mathrm{par})}) > 0$ generically, so no
parameter lies in $\ker J^{(\mathrm{par})}$.
Checkable sufficient conditions on the observable and ansatz
under which this holds are given in
Remark~\ref{rem:sigma_min_conditions} below.

\textbf{Step 4: Null space growth for $P > 2E+1$.}
For $P > 2E+1$, $\dim(\ker J^{(\mathrm{par})}) > 0$
by Step~2, so $\sigma_{\min}(J^{(\mathrm{par})}) = 0$
and parameters begin accumulating in $\ker J^{(\mathrm{par})}$.
By contrast, for serial circuits
$\dim(\ker J^{(\mathrm{ser})})/P \to 1$ across the same range,
since phase-locking confines all columns to $\mathcal{T}_L$
regardless of $P$.
\end{proof}

\begin{remark}[Checkable conditions for 
  $\sigma_{\min}(J^{(\mathrm{par})}) > 0$]
\label{rem:sigma_min_conditions}
The real-analyticity argument (Appendix~\ref{app:cor}, Step~3)
establishes $\sigma_{\min}(J^{(\mathrm{par})}) > 0$ generically
under the following checkable sufficient conditions:
\begin{enumerate}
  \item[(C1)] \textbf{Parameter count}: $P \leq 2E+1$, so the
    map $\Phi:\mathbb{R}^P \to \mathbb{R}^{2E+1}$ goes from a
    lower- to higher-dimensional space.
  \item[(C2)] \textbf{Non-trivial observable}: $M$ has at least
    two distinct eigenvalues (e.g.\ $M = Z_{N-1}$ as used
    throughout), so $\langle 0 | U^\dagger M U | 0 \rangle$
    is non-constant in $\boldsymbol{\theta}$.
  \item[(C3)] \textbf{Circuit expressivity}: the ansatz
    generates a sufficiently rich Lie algebra.
    For the Rot-CNOT ladder used in our experiments,
    the generators span $\mathfrak{su}(2^N)$
    (universal gate set~\citep{brylinski2002universal}),
    so $\Phi$ is non-constant on every open set and
    $\det(J^{(\mathrm{par})\top}J^{(\mathrm{par})})$
    is not identically zero.
\end{enumerate}
Under (C1)--(C3), the zero set of
$q(\boldsymbol{\theta}) = \det(J^{(\mathrm{par})\top}J^{(\mathrm{par})})$
has Lebesgue measure zero by real-analyticity
\citep[Prop.~2.2.7]{bochnak1998real}, so
$\sigma_{\min}(J^{(\mathrm{par})}) > 0$ for almost all
$\boldsymbol{\theta} \in \mathbb{R}^P$.
Degenerate exceptions (where $\sigma_{\min} = 0$ despite
$P \leq 2E+1$) include: $M \propto I$ (trivial output),
circuits with a continuous symmetry fixing all Fourier
coefficients, or initializations in a measure-zero
symmetry-protected submanifold.
For the Rot-CNOT ladder with $M = Z_{N-1}$, none of
these apply generically.
\end{remark}

\section{Serial-Parallel Separation Theorem}
\label{app:separation}

\begin{theorem}[Serial-parallel conditioning separation]
  \label{thm:separation}
  Fix encoding budget $E = NL$ and equal parameter distribution
  $P_n = P/N$ for the parallel case.
  \begin{enumerate}
    \item[(i)] \textbf{Serial} ($N=1$, $L=E$):
      $\mathrm{rank}(J^{(\mathrm{ser})})\leq 2L+1$
      independently of $P$, and
      $\dim(\ker J^{(\mathrm{ser})})/P \to 1$
      as $P \to \infty$ at fixed $L$: an ever-growing fraction
      of parameters receive identically zero gradient signal.
    \item[(ii)] \textbf{Parallel product ansatz} ($N$ qubits,
      $L=E/N$, no inter-qubit entanglement, equal $P_n=P/N$):
      for $P \leq 2E+1$, $\sigma_{\min}(J^{(\mathrm{par})})>0$
    generically and no parameter lies in $\ker J^{(\mathrm{par})}$;
    for $P > 2E+1$, null directions appear and
    $\dim(\ker J^{(\mathrm{par})})/P \to 1$ as $P \to \infty$.
  \end{enumerate}
  For entangling ansätze, the phase-locking argument extends
  qualitatively; the mechanistic justification is given in
  Remark~\ref{rem:entangling}.
\end{theorem}

\begin{proof}
Claim~(i) is Proposition~\ref{prop:ill_cond}.
Claim~(ii), $P \leq 2E+1$: by the absence of phase-locking
established in Appendix~\ref{app:cor} (Step~1) and the
real-analyticity argument of Step~3,
$\sigma_{\min}(J^{(\mathrm{par})}) > 0$ generically,
so no parameter lies in $\ker J^{(\mathrm{par})}$.
Claim~(ii), $P > 2E+1$: by rank-nullity applied to the
shared spectrum $|\Omega| = 2E+1$,
$\dim(\ker J^{(\mathrm{par})}) \geq P - (2E+1) > 0$,
so $\sigma_{\min}(J^{(\mathrm{par})}) = 0$ and null
directions appear at the same threshold as the serial case.
See Appendix~\ref{app:cor} for details.
\end{proof}

\clearpage
\section{Interpolation Threshold Table}
\label{app:interp_table}

\begin{table}[h]
\centering
\caption{
  Interpolation threshold vs.\ $R^2$ saturation in the
  $\mathrm{tbl}$ sweep ($L = L_{\min}$ fixed,
  $n_{\mathrm{train}}=200$).
  $\mathrm{tbl}_{\mathrm{interp}}$: first $\mathrm{tbl}$
  where $P \geq n_{\mathrm{train}}$.
  Gap $= \mathrm{tbl}_{R^2} - \mathrm{tbl}_{\mathrm{interp}}$.
  Negative gaps indicate saturation before the interpolation
  threshold; zero gap indicates coincidence; positive gaps
  indicate that additional parameters beyond the interpolation
  threshold are needed before gradient directions are
  sufficiently well-covered.
}
\label{tab:interp}
\begin{tabular}{ccccccc}
\toprule
$N$ & $\mathrm{deg}$ & $L_{\min}$ & $P(\mathrm{tbl}{=}1)$ &
$\mathrm{tbl}_{\mathrm{interp}}$ &
$\mathrm{tbl}_{R^2 \geq 0.95}$ & Gap \\
\midrule
1 & 12 & 12 &  39 & 6 &  3 & $-3$ \\
2 & 20 & 10 &  66 & 4 &  4 & $\phantom{+}0$ \\
4 & 28 &  7 &  96 & 3 &  6 & $+3$ \\
6 &  6 &  1 &  36 & 6 &  4 & $-2$ \\
\bottomrule
\end{tabular}
\end{table}

The $N=1$ gap of $-3$ is noteworthy: saturation occurs
\emph{before} the interpolation threshold because at
$L_{\min}=12$, the $2L+1=25$ effective Jacobian directions
are already well-covered by $P=117$ parameters at tbl$=3$,
so gradient-based optimization succeeds before the classical
interpolation threshold is reached.
The $N=6$ gap of $-2$ reflects a similar effect: with
$L_{\min}=1$ and a degree-6 target well within the circuit's
accessible spectrum, the $2NL_{\min}+1=13$ Fourier directions
are already well-matched at $P=144$ (tbl$^*=4$), again before
the classical interpolation threshold $P \geq 200$ is reached.
By contrast, the positive gap for $N=4$ ($+3$) indicates that
additional parameters beyond the interpolation threshold are
needed before the optimizer can sufficiently cover the
available gradient directions.

\section{Parameter Efficiency Table}
\label{app:efficiency_table}

Table~\ref{tab:efficiency} quantifies the parameter efficiency of 
the FM route vs.\ the trainable blocks route to reach mean 
$R^2 \geq 0.95$, starting from 
$P_{\mathrm{base}} = (L_{\min}+1) \cdot N \cdot 3$.

\begin{table}[h]
\centering
\caption{
  Parameter efficiency of the FM route vs.\ the trainable
  blocks route to reach mean $R^2 \geq 0.95$, starting from
  $P_{\mathrm{base}} = (L_{\min} +1) \cdot N \cdot 3$.
  $L^*$: smallest $L$ achieving mean $R^2 \geq 0.95$ along
  the FM route (tbl$=1$ fixed).
  tbl$^*$: smallest tbl achieving mean $R^2 \geq 0.95$ along
  the trainable blocks route ($L = L_{\min}$ fixed).
  Ratio $= P_{\mathrm{tbl}} / P_{\mathrm{FM}}$.
}
\label{tab:efficiency}
\begin{threeparttable}
\begin{tabular}{cccccccc}
\toprule
$N$ & deg & $P_{\mathrm{base}}$ &
$L^*$ & $P_{\mathrm{FM}}$ &
tbl$^*$ & $P_{\mathrm{tbl}}$ &
Ratio \\
\midrule
1 &  12 &  39 & 20 &  63 &  3 & 117 & $1.9\times$ \\
2 &  20 &  66 & 19 & 120 &  4 & 264 & $2.2\times$ \\
4 &  28 &  96 & 22 & 276 &  6 & 576 & $2.1\times$ \\
6 &   6 &  36 &  4 &  90 &  4 & 144 & $1.6\times$ \\
\bottomrule
\end{tabular}
\end{threeparttable}
\end{table}

\clearpage
\section{Experimental Details}
\label{app:experiments}

\paragraph{Circuit architecture.}
All experiments use the Rot ansatz with CNOT ladder
entanglement: each trainable block consists of a
$\mathrm{Rot}(\phi, \theta, \omega)$ gate on wire 0
followed by alternating CNOT and Rot gates on wires
$1, \ldots, N-1$.
Data encoding uses unary $\mathrm{RX}(x)$ gates on
all qubits.
The observable is $M = Z_{N-1}$, the Pauli-$Z$ operator
on the final qubit.
All circuits are implemented in PennyLane~\citep{bergholm2018pennylane}
with a JAX backend~\citep{jax2018github}.

\paragraph{Optimizer and initialization.}
All models are trained with Adam~\citep{kingma2014adam}
at learning rate $\eta = 10^{-3}$ for 5,000 steps.
Parameters are initialized uniformly in $[0, 2\pi]$,
matching standard practice for variational quantum
circuits~\citep{schuld2021effect}.
Each experiment is repeated over 10 random seeds per
target function and 5 target functions, giving 50
independent runs per architecture and degree.

\paragraph{Target functions.}
Synthetic targets are random Fourier series of fixed
degree $d$:
\begin{equation}
  f^*(x) = a_0 + \sum_{k=1}^{d}
  \left[ a_k \cos(kx) + b_k \sin(kx) \right],
\end{equation}
with coefficients $a_k, b_k \sim \mathcal{N}(0,1)$
and $f^*$ normalized to unit variance.
Training uses $n_{\mathrm{train}} = 200$ points
drawn uniformly from $[0, 2\pi]$; evaluation uses
$n_{\mathrm{test}} = 100$ held-out points.

\paragraph{Degree selection.}
The target degree for each architecture is chosen as the
smallest degree at which the minimal configuration
($L = L_{\min}$, tbl$=1$) fails to reach $R^2 \geq 0.95$:
deg$=12$ for $N=1$, deg$=20$ for $N=2$, deg$=28$ for $N=4$,
and deg$=6$ for $N=6$.
This places each experiment in a regime where the starting
configuration is genuinely insufficient and adding parameters
along either route is necessary to reach saturation.

\paragraph{Jacobian diagnostics.}
The coefficient matching Jacobian $J \in \mathbb{R}^{(2E+1) \times P}$
is computed by differentiating through a Fourier coefficient extraction
function: for each architecture $(N, L)$, the circuit output is
evaluated on a DFT grid of $n_{\mathrm{diag}} = 2N(L+1)+1$
equally spaced points $x_k = 2\pi k/n_{\mathrm{diag}}$,
$k = 0,\ldots,n_{\mathrm{diag}}-1$, and the real Fourier representatives
$\tilde{\mathbf{c}}(\boldsymbol{\theta})$
(Definition~\ref{def:jacobian}) are extracted via FFT.
The Jacobian $J = \partial\tilde{\mathbf{c}}/\partial\boldsymbol{\theta}$
is then computed using \texttt{jax.jacrev}.
The grid size $n_{\mathrm{diag}} = 2N(L+1)+1$ slightly exceeds
the theoretical spectrum size $2E+1 = 2NL+1$; the $2N$ additional
DFT points lie outside the circuit's accessible spectrum and
produce near-zero Jacobian rows, so the rank is bounded by
$\min(2E+1, P)$ with no artificial truncation.
Singular values are computed via \texttt{jnp.linalg.svd};
hard rank counts all singular values above a threshold of $10^{-10}$.
The Jacobian QFIM
$\hat{\mathcal{F}} = \mathbf{J}_{\mathrm{data}}^\top
\mathbf{J}_{\mathrm{data}} / n$
used as the spectral diagnostic in Section~\ref{sec:op} is computed
analogously at convergence for the FM vs.\ trainable blocks experiments.

\paragraph{Phase-locking experiment (Figure~\ref{fig:phase_locking}).}
To validate Lemma~\ref{lemma:serial_coupling} directly, we compute
the gradient coefficient $g_j(x) = \langle\psi^\perp|\partial_j\psi
\rangle \in \mathbb{C}$ as $x$ sweeps $[0, 2\pi]$ at a fixed random
initialization, for a serial circuit ($N=1$, $L=4$, $P=15$) and a
parallel circuit ($N=2$, $L=2$, $P=18$).
For each parameter $j$, $g_j(x)$ is computed via
\texttt{jax.jacrev} applied to the complex circuit output before
taking the expectation value, normalized by $\max_x |g_j(x)|$,
and plotted as a parametric curve in $\mathbb{C}$.
Phase-locking is quantified as $\mathrm{Im}(g_j/g_k)$ averaged
over $x$ for all pairs $(j,k)$; a value of $0$ indicates
perfect phase-locking (Lemma~\ref{lemma:serial_coupling}),
while nonzero values indicate independent phase trajectories.

\paragraph{Rank ceiling experiment
(Figure~\ref{fig:scaling_separation}).}
To validate Proposition~\ref{prop:ill_cond}(i),
we sweep encoding budget $E \in \{2, 4, 6, 8, 10, 12, 16\}$
across architectures $N \in \{1, 2, 4\}$ at matched
parameter counts $P \approx 3E$.
For each $(E, N)$ pair, 100 random initializations
are evaluated before training.
The coefficient matching Jacobian $J \in \mathbb{R}^{(2E+1)\times P}$
is computed on a DFT grid of $n_{\mathrm{diag}} = 2N(L+1)+1$
equally spaced points (see Jacobian diagnostics paragraph above),
with $\mathrm{rank}(J)$ computed from the resulting singular
values before training.
The serial architecture uses $N=1$, $L=E$; parallel
architectures use $L = E/N$ FM layers per qubit.
For $N=4$, only $E$ divisible by 4 are used to ensure
integer $L$; non-integer rounding would reduce the actual
encoding budget and produce an artifactual gap below the ceiling.

\paragraph{FM vs.\ trainable blocks experiment
(Figure~\ref{fig:FM_vs_tbl}).}
To validate the efficiency claims of Section~\ref{sec:op},
we run two sweeps for each architecture:
the FM route sweeps $L$ from $L_{\min}$ to $L_{\max}$
with $\mathrm{tbl}=1$ fixed, and the trainable blocks route
sweeps $\mathrm{tbl}$ from 1 to $\mathrm{tbl}_{\max}$ with
$L = L_{\min}$ fixed.
The ranges are:
\begin{center}
\begin{tabular}{ccccc}
\toprule
$N$ & $L_{\min}$ & $L_{\max}$ & $\mathrm{tbl}_{\max}$ \\
\midrule
1 & 12 & 89 & 10 \\
2 & 10 & 40 &  4 \\
4 &  7 & 32 &  8 \\
6 &  1 & 21 & 10 \\
\bottomrule
\end{tabular}
\end{center}
$L_{\max}$ is chosen to extend well beyond the saturation
threshold $L^*$; $\mathrm{tbl}_{\max}$ is chosen to include
at least one configuration beyond tbl$^*$.
Both sweeps use $N \in \{1, 6\}$ in the main text
and $N \in \{2, 4\}$ in Appendix~\ref{app:n2_n4_fm_tbl}.
The Jacobian QFIM is evaluated at convergence,
averaged over all runs, to produce the normalized
eigenvalue spectra of Figure~\ref{fig:FM_vs_tbl}.

\paragraph{Real-world validation experiments
(Figures~\ref{fig:realworld_deg12}
and~\ref{fig:realworld_deg24}).}
Both experiments use the Nottingham temperature
dataset~\citep{hipel1994time} with $n_{\mathrm{train}}=200$
and $n_{\mathrm{test}}=40$ points by fixed random split,
time index mapped to $[0,2\pi]$, target scaled to $[-1,1]$.
\emph{Experiment 1} ($E = N \times L_{\min} = 12$):
$L_{\min} \in \{12, 6, 3, 2\}$ for $N \in \{1,2,4,6\}$.
\emph{Experiment 2} ($E = N \times L_{\min} = 24$):
$L_{\min} \in \{24, 12, 6, 4\}$ for $N \in \{1,2,4,6\}$.
For both experiments, FM sweep step sizes are 6, 3, 2, 1
for $N=1,2,4,6$ respectively (approximately constant
parameter increment ${\sim}18$ per step); the tbl sweep
uses tbl$=1,\ldots,20$ with $L=L_{\min}$ fixed.
Optimizer, initialization, and Jacobian diagnostic settings
are identical to the synthetic experiments.

\clearpage
\section{Ruling Out Barren Plateaus}
\label{app:bp}

Barren plateaus~\citep{mcclean2018barren} are characterized by
exponential decay of gradient variance with circuit depth.
Figure~\ref{fig:app_bp} shows that gradient variance is
approximately constant within each architecture as parameter
count increases, with no sign of this decay.
Crucially, the serial architecture $N=1$ has the largest
gradient variance ($\approx 10^{-2}$) yet fails to train
(cf.\ Figure~\ref{fig:intro}), directly ruling out gradient
vanishing as the primary failure mode.

\begin{figure}[h]
  \centering
  \includegraphics[width=\linewidth]{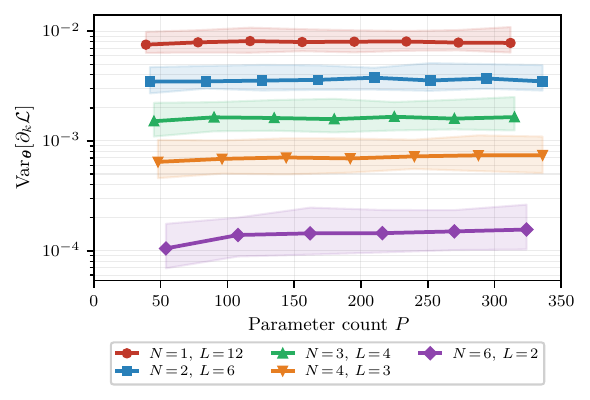}
  \caption{Gradient variance
    $\mathrm{Var}_{\boldsymbol{\theta}}[\partial_k \mathcal{L}]$
    vs.\ parameter count $p$ for degree-12 targets, separated
    by architecture.
    Median over random initializations; shaded bands show IQR.
    Variance is approximately constant within each architecture,
    with no sign of depth-dependent exponential decay.}
  \label{fig:app_bp}
\end{figure}

\clearpage
\section{Jacobian Rank Ceiling: Empirical Validation}
\label{app:rank_saturation}

Figure~\ref{fig:scaling_separation} validates
Proposition~\ref{prop:ill_cond}(i) empirically across
encoding budgets $E \in \{2,4,6,8,10,12,16\}$ and
architectures $N \in \{1,2,4\}$ at matched parameter
counts $P \approx 3E$. All architectures reach the rank
ceiling $2E+1$ at matched $P$, confirming that the ceiling
is tight: at $P > 2E+1$, a growing fraction $(P-(2E+1))/P$
of parameters lie in $\ker J$ regardless of architecture.
For $N=4$, only $E$ divisible by 4 are used to ensure
integer $L = E/N$; non-integer rounding would silently
reduce the actual encoding budget and produce an artifactual
gap below the ceiling.

\begin{figure}[h]
  \centering
  \includegraphics[width=0.5\textwidth]{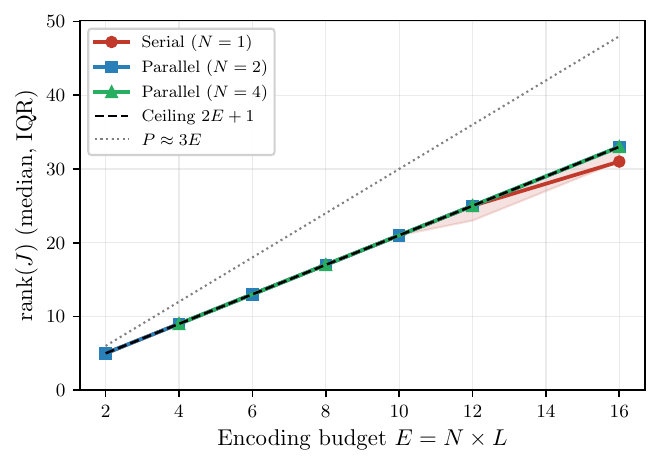}
  \caption{
  \textbf{Jacobian rank saturates at the $2E+1$ ceiling for
  all architectures at matched $P \approx 3E$.}
  Median $\mathrm{rank}(J)$ with IQR across 100 random
  initializations at matched parameter counts $P \approx 3E$
  (dotted gray). All architectures track the theoretical
  ceiling $2E+1$ (dashed black), confirming
  Proposition~\ref{prop:ill_cond}(i): the ceiling is tight
  and at $P > 2E+1$ a fraction $(P-(2E+1))/P$ of parameters
  lie in $\ker J$ regardless of architecture.
  For $N=4$, only $E$ divisible by 4 are shown to ensure
  integer $L = E/N$.
  }
  \label{fig:scaling_separation}
\end{figure}

\clearpage
\section{Bilinear Factorization: Circuit Diagram}
\label{app:bilinear_fig}

Figure~\ref{fig:bilinear} illustrates the structural difference
between serial and parallel coefficient matching equations at
encoding budget $E=2$, complementing the algebraic comparison
of Section~\ref{subsec:concrete}.
The contrast between chain coupling (serial) and bilinear
separation (parallel) is visible directly from the circuit
structure: in the serial case all three parameter matrices
are coupled through the chain of encoding layers, while in
the parallel case $W^{(1)}$ and $W^{(2)}$ enter through
entirely separate factors.

\begin{figure}[h!]
\centering
\tikzset{
  every node/.style={},
  ann/.style={font=\footnotesize},
  lbl/.style={font=\small\bfseries},
}
\begin{tikzpicture}[scale=1.0]

\node[anchor=north west] (serial) at (-1.5, 0) {
  \begin{quantikz}[column sep=0.3em, row sep=0.3em]
    \lstick{$|0\rangle$} &
    \gate[style={fill=blue!12}]{W^{(1)}} &
    \gate[style={fill=gray!15}]{S(x)} &
    \gate[style={fill=blue!12}]{W^{(2)}} &
    \gate[style={fill=gray!15}]{S(x)} &
    \gate[style={fill=blue!12}]{W^{(3)}} &
    \meter{M}
  \end{quantikz}
};

\node[ann, anchor=north] at (2.5, -3.5)
  {$c_{+2}^{\mathrm{ser}} = 
  -2\,\alpha_1\,\alpha_2^2\,\bar\beta_1\,\beta_3\bar\alpha_3$};

\node[anchor=north west] (parallel) at (6.5, 0) {
  \begin{quantikz}[column sep=0.3em, row sep=0.5em]
    \lstick{$|0\rangle$} &
    \gate[style={fill=blue!12}]{W^{(1)}_1} &
    \gate[style={fill=gray!15}]{S(x)} &
    \gate[style={fill=blue!12}]{W^{(2)}_1} &
    \meter{M} \\
    \lstick{$|0\rangle$} &
    \gate[style={fill=blue!12}]{W^{(1)}_2} &
    \gate[style={fill=gray!15}]{S(x)} &
    \gate[style={fill=blue!12}]{W^{(2)}_2} &
    \qw
  \end{quantikz}
};

\draw[dashed, gray!50, thick] (6.2, 0.35) -- (6.2, -2.5);

\node[ann, anchor=south] at (2.5, 0.15) {\textbf{Chain coupling}};
\node[lbl, anchor=south] at (8.8, 0.15) {\textbf{Bilinear separation}};

\draw[decorate, decoration={brace, amplitude=4pt, mirror}, thick]
  (-1.0, -2.5) -- (6.0, -2.5);
\node[ann, anchor=north] at (2.5, -2.75)
  {degree $2L{+}2{=}6$, all matrices coupled};

\draw[decorate, decoration={brace, amplitude=4pt, mirror}, thick]
  (7.25, -2.5) -- (8.5, -2.5);
\node[ann, anchor=north] at (8.0, -2.75) {$W^{(1)}$};

\draw[decorate, decoration={brace, amplitude=4pt, mirror}, thick]
  (9.5, -2.5) -- (10.75, -2.5);
\node[ann, anchor=north] at (10.25, -2.75) {$W^{(2)}$};

\node[ann, anchor=north] at (9.0, -3.5)
  {$c_{+2}^{\mathrm{par}} =
    \underbrace{v_1\bar{v}_4}_{W^{(1)}}
    \cdot
    \underbrace{\widetilde{M}_{4,1}}_{W^{(2)}}$};

\end{tikzpicture}
\caption{
  Structural comparison of coefficient matching equations
  at encoding budget $E=2$.
  \textbf{Left (serial, $N=1$, $L=2$):} $c_{+2}^{\mathrm{ser}}$
  couples all three trainable parameter matrices simultaneously
  in a degree-6 polynomial --- every gradient update entangles
  all parameter blocks.
  \textbf{Right (parallel, $N=2$, $L=1$):} $c_{+2}^{\mathrm{par}}$
  factors into a $W^{(1)}$-dependent term and a
  $W^{(2)}$-dependent term with no cross-coupling
  (Proposition~\ref{prop:bilinear}), consistent with the independent phase trajectories
of Corollary~\ref{cor:parallel}.
  Blue boxes: trainable ansatz layers.
  Gray boxes: data-encoding layers $S(x)$.
  CNOT entanglement omitted for clarity; see
  Remark~\ref{rem:entangling}.
}
\label{fig:bilinear}
\end{figure}
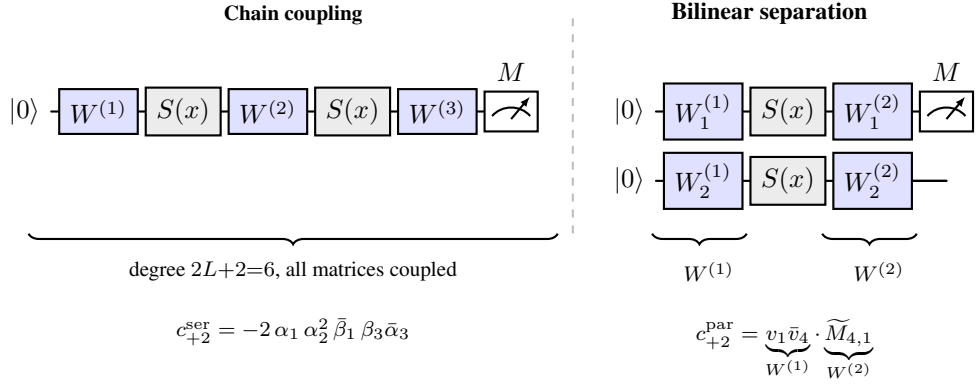

\clearpage
\section{Real-World Validation: Nottingham Temperature Dataset}
\label{app:realworld}

We validate the structural predictions of
Sections~\ref{sec:theory} and~\ref{sec:op} on the Nottingham
temperature dataset~\citep{hipel1994time}, a canonical
univariate time series of 240 monthly mean air temperatures
recorded at Nottingham Castle, England, from 1920 to 1939.
We map the time index to $[0, 2\pi]$, scale the target to
$[-1, 1]$, and use $n_{\mathrm{train}} = 200$ and
$n_{\mathrm{test}} = 40$ points by fixed random split.
Two experiments are reported using different encoding budgets
$E = N \times L_{\min}$ to probe different regimes:
$E=12$ (Figure~\ref{fig:realworld_deg12}) places all
architectures in the expressivity-limited regime where the
model spectrum is insufficient at $L_{\min}$;
$E=24$ (Figure~\ref{fig:realworld_deg24}) gives each
architecture sufficient expressivity at $L_{\min}$ and
isolates the parameter efficiency question.

\subsection*{Experiment 1: $E = N \times L_{\min} = 12$
(Expressivity Regime)}

\begin{figure}[h!]
  \centering
  \includegraphics[height=0.88\textheight,keepaspectratio]{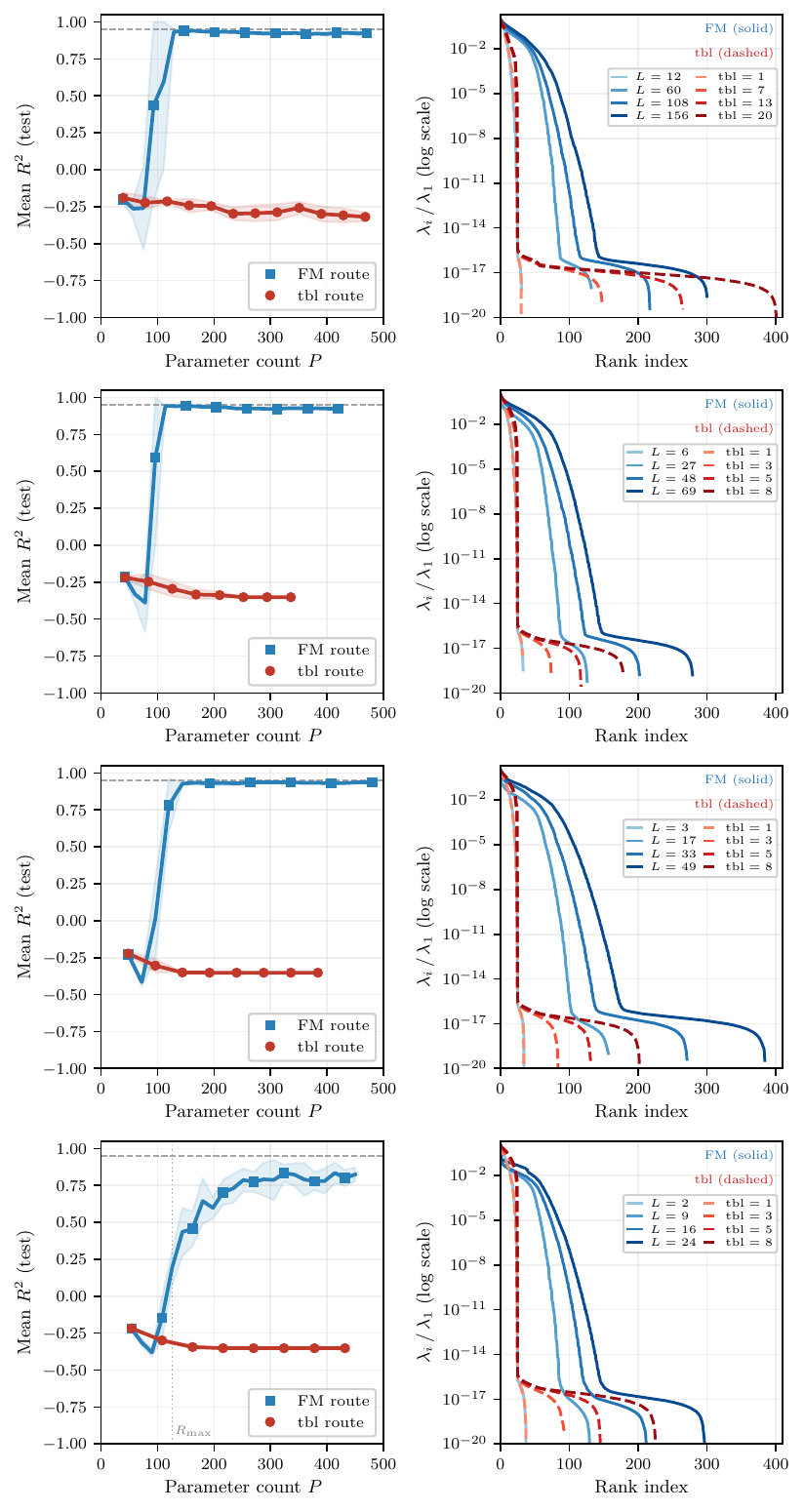}
  \caption{
    \textbf{Real-world validation on the Nottingham temperature
    dataset ($E = N \times L_{\min} = 12$,
    $n_{\mathrm{train}}=200$).}
    Each row shows one architecture ($N \in \{1,2,4,6\}$).
    \textbf{Left panels:} mean $R^2$ (test) $\pm$ std
    vs.\ parameter count $P$; dotted vertical line marks
    $R_{\max}=126$ for $N=6$ only.
    \textbf{Right panels:} normalized Jacobian QFIM spectra
    $\lambda_i/\lambda_1$ at selected $L$ (FM, solid blue)
    and tbl (trainable blocks, dashed red) values.
    The trainable blocks route fails categorically across all
    four architectures ($R^2 < 0$ throughout), while the FM
    route reaches $R^2 \geq 0.95$ for $N \in \{1,2,4\}$
    and approaches but does not fully reach saturation for
    $N=6$.
    The spectral knee shifts rightward along the FM route and
    remains frozen along the tbl route in all four cases,
    consistent with Figure~\ref{fig:FM_vs_tbl}.
  }
  \label{fig:realworld_deg12}
\end{figure}

With $E = N \times L_{\min} = 12$, the encoding budgets are
$L_{\min} \in \{12, 6, 3, 2\}$ for $N \in \{1,2,4,6\}$
respectively.
Unlike the synthetic experiments, the Nottingham spectrum
is not perfectly bandlimited: significant spectral content
lies beyond $L_{\min}$ coverage, so FM layers provide a
dual benefit --- expanding Jacobian rank coverage \emph{and}
expanding expressivity --- while the trainable blocks route
can only improve conditioning within the fixed rank ceiling
without expanding the accessible function class.

The trainable blocks route is stuck at $R^2 < 0$ across the
full parameter range for all four architectures.
For $N=2$ and $N=4$, which successfully trained on synthetic
targets via the tbl route, the failure here is not gradient
starvation but an expressivity ceiling: the Nottingham
spectrum extends beyond the function class accessible at
fixed $L_{\min}$, and adding trainable blocks cannot expand
it.
For $N=1$, structural gradient starvation compounds the
expressivity failure.
The FM route succeeds for $N \in \{1,2,4\}$ by raising $L$,
simultaneously expanding the accessible spectrum and the
Jacobian rank ceiling.
For $N=6$, the FM route improves monotonically but plateaus
around $R^2 \approx 0.80$; the structural interpretation
is given in Experiment~2 below.

\clearpage
\subsection*{Experiment 2: $E = N \times L_{\min} = 24$
(Parameter Efficiency Regime)}

\begin{figure}[h!]
  \centering
  \includegraphics[height=0.88\textheight,keepaspectratio]{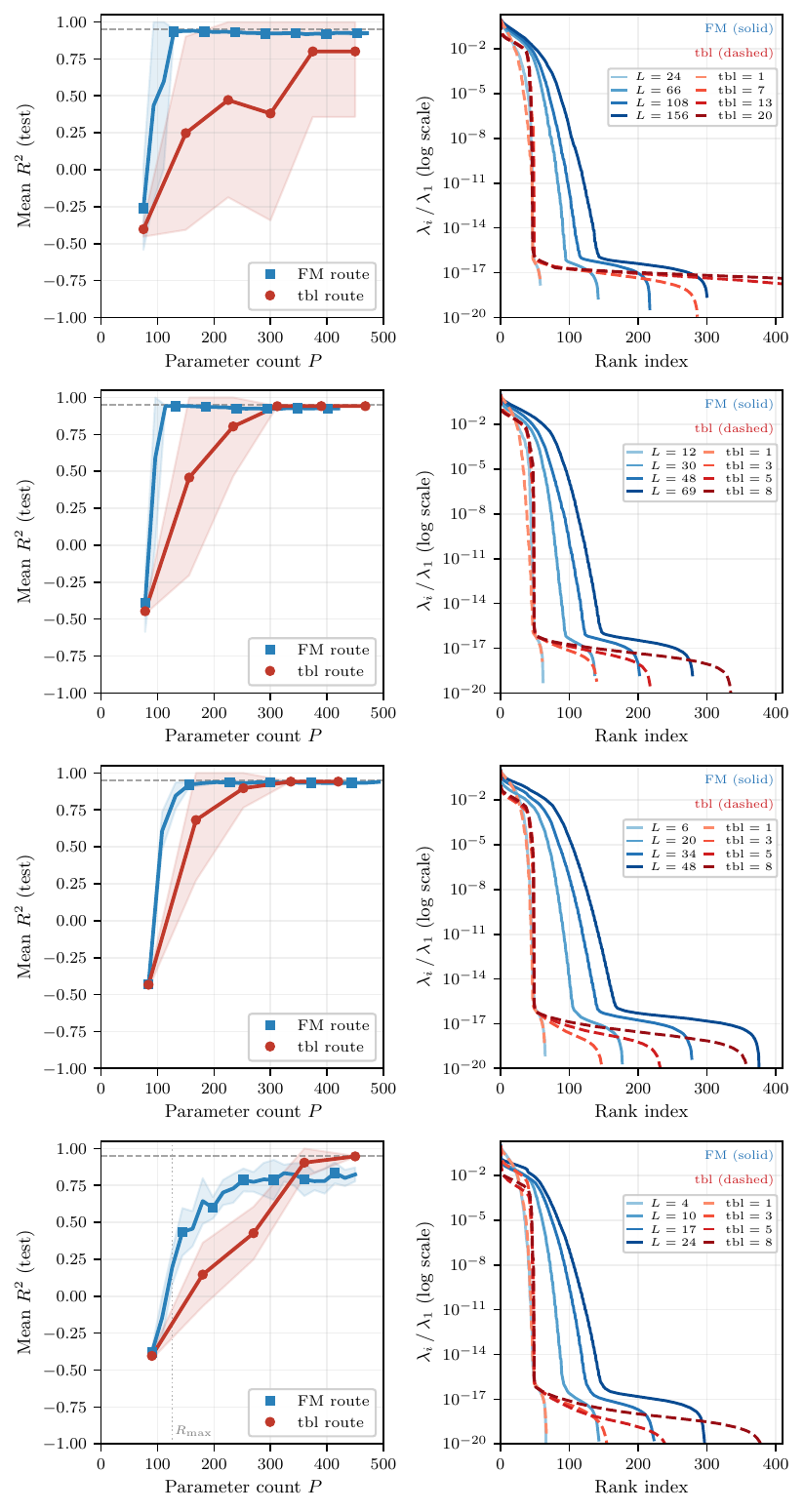}
  \caption{
    \textbf{Real-world validation on the Nottingham temperature
    dataset ($E = N \times L_{\min} = 24$,
    $n_{\mathrm{train}}=200$).}
    Layout identical to Figure~\ref{fig:realworld_deg12}.
    $L_{\min} \in \{24, 12, 6, 4\}$ for $N \in \{1,2,4,6\}$
    respectively; $R_{\max}=126$ line shown for $N=6$ only.
    For $N \in \{1,2,4\}$, the FM route reaches $R^2 \geq 0.94$
    with $1.9$--$2.7\times$ fewer parameters than the trainable
    blocks route, consistent with the synthetic results of
    Table~\ref{tab:efficiency}.
    For $N=6$, the result inverts: the tbl route reaches
    $R^2 \geq 0.95$ at tbl$=5$ ($P=450$) while the FM route
    plateaus around $R^2 \approx 0.83$, consistent with the
    circuit entering the Larocca underparameterization regime
    $P \approx R_{\max}=126$~\citep{larocca2023theory} early
    in the FM sweep.
  }
  \label{fig:realworld_deg24}
\end{figure}

With $E = N \times L_{\min} = 24$, the encoding budgets are
$L_{\min} \in \{24, 12, 6, 4\}$ for $N \in \{1,2,4,6\}$
respectively, giving each architecture sufficient
expressivity at $L_{\min}$ to represent the target spectrum.
This isolates the parameter efficiency question from
expressivity effects.

\paragraph{FM efficiency for $N \in \{1,2,4\}$.}
For $N=1$, the FM route reaches $R^2 \approx 0.94$ at
$L=43$ ($P=129$) while the tbl route reaches comparable
performance at tbl$=7$ ($P=525$), a $4.1\times$ efficiency
ratio.
For $N=2$, the FM route reaches $R^2 \geq 0.94$ at $L=19$
($P=114$) vs.\ tbl$=4$ ($P=312$), a $2.7\times$ ratio.
For $N=4$, FM reaches $R^2 \geq 0.93$ at $L=15$ ($P=180$)
vs.\ tbl$=4$ ($P=336$), a $1.9\times$ ratio.
These efficiency ratios are consistent with the synthetic
results of Table~\ref{tab:efficiency}, confirming that the
FM parameter efficiency advantage generalizes to real-world
data when expressivity is not the binding constraint.

\paragraph{Inversion for $N=6$ and the Larocca boundary.}
For $N=6$, the result inverts: the tbl route reaches
$R^2 \geq 0.95$ at tbl$=5$ ($P=450$) while the FM route
plateaus around $R^2 \approx 0.83$ despite each FM layer
expanding the accessible spectrum by $N=6$ encoding units.
The right panel for $N=6$ shows the spectral knee growing
less with each FM layer than for $N \in \{1,2,4\}$, and
eigenvalue mass concentrated at lower rank indices.
This is consistent with the $N=6$ circuit entering the
Larocca underparameterization regime at $P \approx R_{\max}
= 126$~\citep{larocca2023theory} early in the FM sweep:
at $L_{\min}=4$, $P_{\mathrm{base}} = 90$, and the FM
route crosses $R_{\max}$ at $L=5$ ($P=108$).
Beyond this point, adding FM layers expands the model
spectrum and requires control of additional Fourier
coefficients, increasing the parameter demand faster than
the added parameters satisfy it.
The tbl route avoids this by keeping $L=L_{\min}=4$ fixed
and reaching saturation through the classical interpolation
mechanism, unaffected by the Larocca condition.
Whether this inversion is a general property of the
Larocca-limited regime or specific to the $N=6$ architecture
and Nottingham target is left for future investigation.

\clearpage
\section{Spectral Knee Position: Quantitative Diagnostic}
\label{app:spectral_knee}

Figure~\ref{fig:spectral_knee} quantifies the spectral knee
position --- the number of eigenvalues satisfying
$\lambda_i/\lambda_1 > 10^{-6}$ --- as a function of parameter
count $P$ along both routes for $N=1$ and $N=6$.

For $N=1$ (left panel), the FM route shifts the knee
monotonically from rank $25$ at $L=L_{\min}=12$ to
rank $52$ at $L=89$: each extra FM layer expands the
ceiling $2L+1$ and makes genuinely new Fourier directions
accessible.
The trainable blocks route knee is frozen at rank $24$--$25$, 
matching the theoretical ceiling of $25$, with no monotonic growth ---
confirming that no new Fourier directions become accessible
as trainable blocks are added.

For $N=6$ (right panel), the FM route shifts the spectral knee
monotonically from rank $13$ at $L=L_{\min}=1$ to rank $79$
at $L=21$, and both routes reach $R^2 \geq 0.95$.
The FM route reaches saturation at $L^*=4$ ($P=90$) while the
tbl route reaches saturation at tbl$^*=4$ ($P=144$), giving a
$1.6\times$ efficiency ratio.
The tbl knee is frozen at $12$--$13$, 
matching the theoretical rank ceiling of $2NL_{min}+1=13$
regardless of parameter count, consistent with the
frozen-ceiling pattern seen across all architectures:
the tbl route improves $R^2$ via the classical interpolation
mechanism rather than by accessing new Fourier directions.

\begin{figure}[h!]
  \centering
  \includegraphics[width=\linewidth]{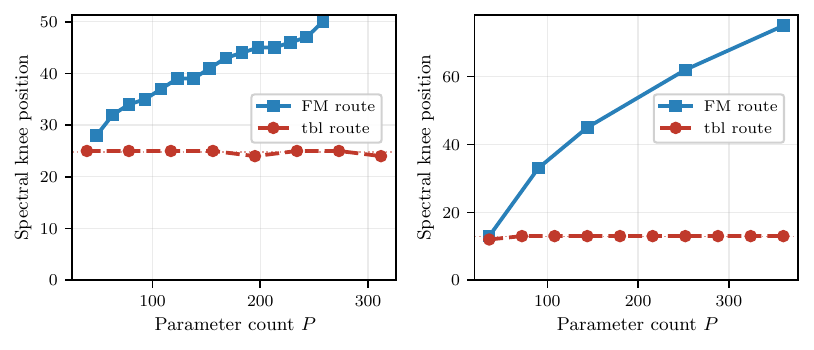}
  \caption{
    \textbf{Spectral knee position vs.\ parameter count $P$
    along the FM route (blue) and trainable blocks route (red).}
    The knee is defined as the number of normalized eigenvalues
    $\lambda_i/\lambda_1$ exceeding $10^{-6}$, a
    threshold-free measure stable across several decades.
    \textbf{Left ($N=1$):} the FM knee grows monotonically
    from rank $25$ at $L_{\min}=12$ to rank $52$ at $L=89$,
    while the trainable blocks knee is frozen at rank
    $24$--$25$, matching the theoretical rank ceiling
    $2L_{\min}+1=25$, regardless of parameter count ---
    confirming the frozen-ceiling claim of
    Section~\ref{subsec:fm_route}.
    \textbf{Right ($N=6$):} the FM knee grows from rank $13$
    to $79$ while the tbl knee is frozen at $12$--$13$,
    matching the theoretical rank ceiling
    $2NL_{\min}+1=13$; both routes reach $R^2 \geq 0.95$.
  }
  \label{fig:spectral_knee}
\end{figure}

\clearpage
\section{FM vs.\ Trainable Blocks: Full Results for
\texorpdfstring{$N=2$}{N=2} and \texorpdfstring{$N=4$}{N=4}}
\label{app:n2_n4_fm_tbl}

Figure~\ref{fig:appendix_FM_tbl_N2_N4} shows the FM and classical
routes for $N=2$, deg$=20$ and $N=4$, deg$=28$ --- the two
architectures in the well-conditioned regime where both routes
reach mean $R^2 \geq 0.95$ within the tested range.
These complement the main text results for $N=1$ and $N=6$
(Figure~\ref{fig:FM_vs_tbl}), which illustrate contrasting
efficiency regimes.

For both $N=2$ and $N=4$, the qualitative picture matches the
theoretical prediction: the FM route reaches saturation with $2.2\times$ fewer
parameters than the trainable blocks route for $N=2$
and $2.1\times$ for $N=4$ (Table~\ref{tab:efficiency}),
and the normalized QFIM eigenvalue
spectra show the spectral knee shifting rightward along the FM
route while remaining fixed along the trainable blocks route.
The frozen-knee pattern is particularly clear for $N=2$, where
the tbl curves collapse visibly at the same knee position across
tbl$=1$\ldots$4$ despite a four-fold increase in parameter count.
For $N=4$ the FM curves are compressed near the left edge of the
rank axis, reflecting the larger state space ($2^N = 16$) which
distributes eigenvalue mass across more directions.
The $N=4$ tbl sweep exhibits non-monotonic $R^2$ behavior:
mean $R^2$ dips from $0.621$ at tbl$=1$ to $0.570$ at tbl$=2$
before recovering monotonically to $0.984$ at tbl$=6$.
The cause of this dip was not investigated further; it does
not affect the structural conclusion that the FM route is
$2.1\times$ more parameter-efficient.

\begin{figure}[h]
  \centering
  \includegraphics[width=\linewidth]{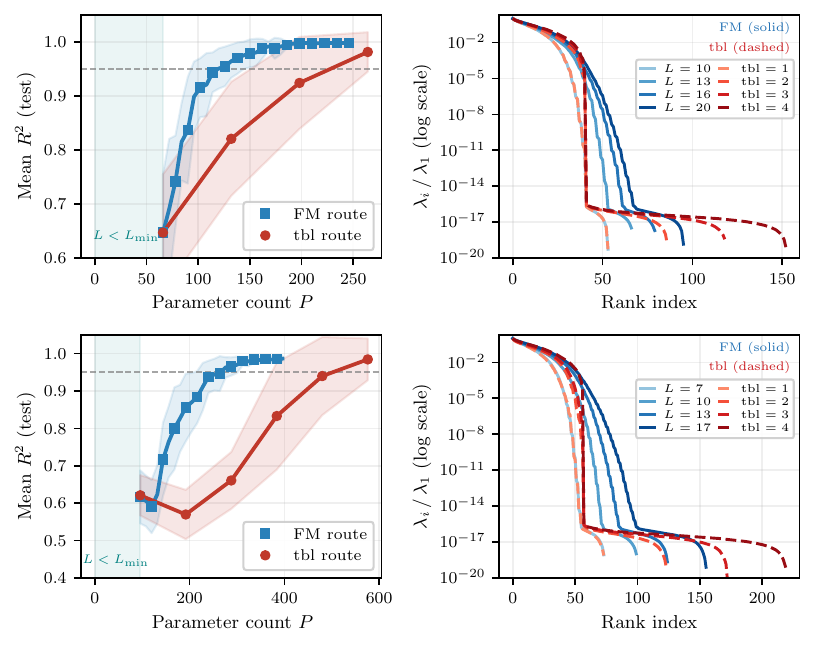}
  \caption{
    \textbf{FM route vs.\ trainable blocks route for $N=2$
    (top row) and $N=4$ (bottom row).}
    \textbf{Left panels:} mean $R^2$ (test) $\pm$ std
    vs.\ parameter count $P$.
    Teal shading: expressivity gap ($L < L_{\min}$).
    Dashed gray line: $R^2 = 0.95$ threshold.
    Both routes reach saturation; the FM route requires
    $2.2\times$ fewer parameters for $N=2$ and $2.1\times$
    for $N=4$ (Table~\ref{tab:efficiency}).
    \textbf{Right panels:} normalized Jacobian QFIM
    spectra $\lambda_i/\lambda_1$ at selected $L$
    values (FM, solid blue) and tbl values (trainable 
    blocks, dashed red).
    The spectral knee generally shifts rightward along the FM
    route and remains fixed along the trainable blocks 
    route, consistent with Figure~\ref{fig:FM_vs_tbl} and
    the frozen-knee claim of
    Section~\ref{subsec:fm_route}.
  }
  \label{fig:appendix_FM_tbl_N2_N4}
\end{figure}

\clearpage
\section{Degree Selection for Efficiency Experiments}
\label{app:degree_justification}

The efficiency comparison in Table~\ref{tab:efficiency} uses
a fixed target Fourier degree for each architecture:
deg$=12$ for $N=1$, deg$=20$ for $N=2$, deg$=28$ for $N=4$,
and deg$=6$ for $N=6$.
Figure~\ref{fig:degree_justification} justifies these choices
by showing mean $R^2$ (test) as a function of target degree
for each architecture, evaluated at matched encoding budget
$E = \mathrm{deg}$ with tbl$=1$.

For $N=2$ and $N=4$, the chosen degrees sit at the onset
of reliable degradation: mean $R^2$ remains near 1.0 for
lower degrees and begins to decline consistently beyond the
chosen value, placing the experiment in a regime where the
architecture is genuinely challenged without already failing.
For $N=1$, deg$=12$ falls within the declining regime ---
consistent with the gradient starvation mechanism, which
causes serial circuits to degrade at lower degrees than
parallel architectures of the same encoding budget.
For $N=6$, deg$=6$ is the lowest tested degree; the minimum
configuration ($L_{\min}=1$, tbl$=1$, $P=36$) lies well below
the geometric threshold $R_{\max} \leq 2^{N+1}-2 = 126$
of~\citet{larocca2023theory}, explaining the low $R^2 \approx 0.56$
at the starting point.
Both routes reach $R^2 \geq 0.95$ once $P$ approaches or
exceeds $R_{\max}$: the tbl route crosses the threshold around
tbl$=4$ ($P=144$) and the FM route at $L^*=4$ ($P=90$), where
the accessible spectrum ($2NL+1=49$) remains well within
$R_{\max}=126$.
The dashed vertical lines mark the chosen degree for each
architecture.

\begin{figure}[h!]
  \centering
  \includegraphics[width=0.75\linewidth]{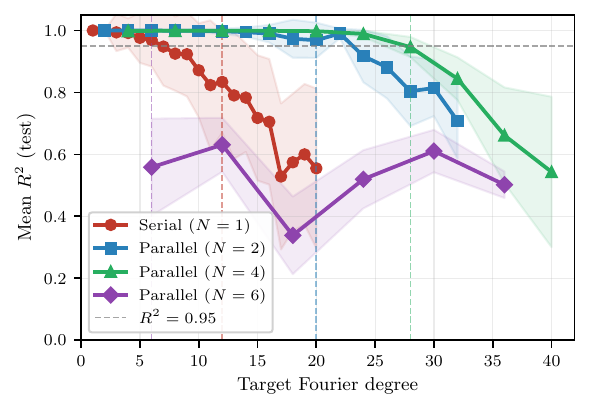}
  \caption{
    Mean $R^2$ (test, $\pm$ std) vs.\ target Fourier degree
    for each architecture, at matched encoding budget
    $E = \mathrm{deg}$ and tbl$=1$.
    Dashed vertical lines mark the degree used in
    Table~\ref{tab:efficiency} for each architecture.
    Dashed gray line: $R^2 = 0.95$ reference.
    For $N=2$ and $N=4$, the chosen degree sits at the
    onset of degradation.
    For $N=1$, deg$=12$ falls within the declining regime,
    consistent with structural gradient starvation.
    For $N=6$, the minimum configuration lies below the geometric
    threshold $R_{\max}=126$~\citep{larocca2023theory}; both routes
    reach $R^2 \geq 0.95$ once $P$ approaches $R_{\max}$.
  }
  \label{fig:degree_justification}
\end{figure}


\end{document}